\documentclass[journal,comsoc]{IEEEtran}
\usepackage[utf8]{inputenc} 
\usepackage[T1]{fontenc}
\usepackage{url}
\usepackage{ifthen}
\usepackage[cmex10]{amsmath}
\usepackage{array} 

\usepackage{enumerate}
\usepackage{amssymb}
\usepackage{amsmath} 

\usepackage{algorithm}
\usepackage{algorithmicx}
\usepackage{algpseudocode}

\usepackage{mathrsfs}
\usepackage{bm}
\usepackage{tikz}
\usetikzlibrary{arrows}
\usepackage{subfigure}
\usepackage{graphicx,booktabs,multirow}
\usepackage{epstopdf}
\usepackage{multirow}
\usepackage{tabularx} 
\usepackage{booktabs}
\usepackage{cite}

\usepackage{pifont}

\definecolor{colorhkust}{RGB}{20,43,140}
\definecolor{colortsinghua}{RGB}{116,52,129}
\definecolor{color1}{RGB}{128,0,0}

\usepackage{amsthm}

\newtheorem{thm}{Theorem}
\newtheorem{lem}{Lemma}
\newtheorem{prop}{Proposition}
\newtheorem{cor}{Corollary}

\theoremstyle{definition}
\newtheorem{defn}{Definition}

\newtheorem{exmp}{Example}

\theoremstyle{remark}

\begin{document}

        \title{Graph Neural Networks for Wireless Communications: From Theory to Practice} 
   
      \author{Yifei Shen, \textit{Graduate Student Member}, \textit{IEEE}, Jun Zhang, \textit{Fellow}, \textit{IEEE}, S.H. Song, \textit{Senior Member}, \textit{IEEE}, and Khaled B. Letaief, \textit{Fellow}, \textit{IEEE}
      \thanks{The materials in this paper were presented in part at the IEEE International Conference on Communications (ICC), 2022 \cite{shen2021neural}. This work was supported by the General Research Fund (Project No. 16210719, 15207220, and 16212120) from the Research Grants Council of Hong Kong.
     
        The authors are with the Department of Electronic and Computer Engineering, Hong Kong University of Science and Technology, Hong Kong (E-mail: yshenaw@connect.ust.hk, \{eejzhang, eeshsong eekhaled\}@ust.hk). Y. Shen is also with Microsoft Research Asia and K. B. Letaief is also with Peng Cheng Lab in Shenzhen.  (The corresponding author is J. Zhang.)}}
        
        \maketitle

\begin{abstract}
Deep learning-based approaches have been developed to solve challenging problems in wireless communications, leading to promising results. Early attempts adopted neural network architectures inherited from applications such as computer vision. They often yield poor performance in large scale networks (i.e., poor scalability) and unseen network settings (i.e., poor generalization). To resolve these issues, graph neural networks (GNNs) have been recently adopted, as they can effectively exploit the domain knowledge, i.e., the graph topology in wireless communications problems. GNN-based methods can achieve near-optimal performance in large-scale networks and generalize well under different system settings, but the theoretical underpinnings and design guidelines remain elusive, which may hinder their practical implementations. This paper endeavors to fill both the theoretical and practical gaps. For theoretical guarantees, we prove that GNNs achieve near-optimal performance in wireless networks with much fewer training samples than traditional neural architectures. Specifically, to solve an optimization problem on an $n$-node graph (where the nodes may represent users, base stations, or antennas), GNNs' generalization error and required number of training samples are $\mathcal{O}(n)$ and $\mathcal{O}(n^2)$ times lower than the unstructured multi-layer perceptrons. For design guidelines, we propose a unified framework that is applicable to general design problems in wireless networks, which includes graph modeling, neural architecture design, and theory-guided performance enhancement. Extensive simulations, which cover a variety of important problems and network settings, verify our theory and the effectiveness of the proposed design framework.
\begin{IEEEkeywords}
Graph neural networks, wireless communication, message passing, provably approximately correct learning, deep learning.
\end{IEEEkeywords} 
\end{abstract}

\section{Introduction}
\subsection{Motivation}
Deep learning has recently emerged as a revolutionary technology to solve various problems in wireless communication systems, e.g., radio resource management \cite{sun2018learning,shen2020graph}, channel estimation \cite{he2020model}, and joint source-channel coding \cite{shao2021learning}. There are two different approaches to developing deep learning-based methods. The first one is the data-driven approach \cite{sun2018learning,liang2018towards,shen2020graph}, which replaces conventional building blocks with a neural network and directly learns the optimal input-output mapping of the problem. The second is the model-driven approach \cite{he2019model,he2020model}, which replaces some policies in a classic algorithm with a neural network. For both paradigms, one essential design component is the underlying neural architecture, which governs the training and generalization performance. 

Early attempts took a plug-and-play approach and adopted neural architectures inherited from applications such as computer vision, e.g., fully connected multi-layer perceptrons (MLPs) or convolutional neural networks (CNNs) \cite{sun2018learning,shen2018lora,liang2018towards}. Although these classic architectures achieve near-optimal performance and fast execution on small-scale wireless networks, the performance is severely degraded when the number of users becomes large. For example, for the data-driven deep learning-based beamforming with CNNs, the performance is near-optimal for a two-user network, while an $18\%$ gap to the classic algorithm exists for a $10$-user network \cite{ma2021neural}. Moreover, traditional neural architectures generalize poorly when the key parameters in the test setting are different from that in the training dataset. For example, for the access point selection problem, the performance degradation of MLP-augmented branch-and-bound can be more than $50\%$ when the user number or signal-to-noise ratio (SNR) in the test dataset is slightly different from that in the training dataset \cite{shen2018lora}. 5G and beyond networks typically have densely deployed access points, hundreds of clients, and dynamically changing client numbers and SNR\cite{letaief2019roadmap,letaief2022Edge}, making it highly ineffective to apply MLP-based or CNN-based methods.

To improve scalability and generalization, a promising direction is to design neural network architectures that are specialized for wireless networks. One recent attempt is to employ graph neural networks (GNNs) that can exploit the graph topology of wireless networks \cite{lee2019graph,eisen2019optimal,shen2020graph,jiang2020learning,kosasih2022graph,chowdhury2021unfolding,lee2021learning,he2022graph,he2021Overview,zhou2022Multi}. GNN-based methods have achieved promising results in applications including resource management \cite{eisen2019optimal,shen2020graph}, end-to-end communication  \cite{jiang2020learning,zhou2022Multi}, and MIMO detection \cite{kosasih2022graph,he2022graph}. One may refer to a recent survey for more applications \cite{he2021Overview}. For example, for the beamforming problem, a GNN trained on a network with $50$ users is able to achieve near-optimal performance in a larger network with $1000$ users \cite{shen2020graph}. Furthermore, thanks to the parallel execution, GNNs are computationally efficient, and by far the only approach that is able to find the near-optimal beamformer for thousands of users in milliseconds \cite{shen2020graph}. However, despite the empirical successes, the theoretical underpinnings and design guidelines remain elusive, which hinders the practical implementations of GNNs in wireless networks.

The current theoretical understanding of the benefits of GNN-based methods is limited to the permutation invariance property \cite{shen2020graph,eisen2019optimal,guo2021learning,lee2021learning}. This property was proved to be universal in radio resource management problems in \cite{shen2020graph}. Given that GNNs well respect this property while MLPs do not \cite{shen2020graph,eisen2019optimal,guo2021learning}, we may expect that GNNs will outperform MLPs and this has been verified empirically in \cite{shen2020graph,eisen2019optimal,guo2021learning,lee2021learning}. Although these works provide some clues on why GNN architectures are better than MLPs for wireless networks, it is not practically useful. First, it fails to characterize the conditions for GNNs to outperform classic neural architectures. Second, the qualitative results of \cite{shen2020graph,eisen2019optimal,guo2021learning,lee2021learning} cannot be used to guide a better design of GNN-based algorithms, i.e., to find a better GNN architecture. To address these issues, a quantitative performance measure is needed. In this paper, we adopt the provably approximately correct learning (PAC-learning) framework \cite{valiant1984theory}, which is a standard way to characterize the generalization errors with a finite amount of training samples. With this framework, we are able to theoretically compare the generalization performance of GNN-based and MLP-based methods. Our analysis also leads to a principled and general approach to designing neural architectures for GNNs, while existing studies mainly focus on individual application cases \cite{shen2019graph,lee2019graph,guo2021learning,jiang2020learning,chowdhury2021unfolding,zhou2022Multi}.

\subsection{Contributions}
In this paper, we develop a unified framework for applying GNNs to wireless communications. Specifically, our main contributions are summarized as follows:
\begin{enumerate}
    \item We present a general framework for modeling typical problems of wireless communication as graph optimization problems. For illustrations, we apply this framework to power control in device-to-device (D2D) networks, resource allocation in cell-free networks, and hybrid precoding for millimeter-wave (mmWave) networks.
	
    \item To bridge the learning-based algorithms and optimization-based algorithms in wireless networks, we identify an important class of algorithms for solving graph optimization problems, named \emph{distributed message passing} (DMP) algorithms \cite{angluin1980local}. DMP algorithms unify many classic optimization-based algorithms (e.g., WMMSE \cite{Shi2011An}, Riemannian gradient \cite{yu2016alternating}, and belief propagation \cite{yedidia2000bethe}), as well as GNN-based algorithms (e.g., \cite{shen2019graph}) for problems in wireless networks. It is proved that GNNs are the most powerful DMP algorithms, i.e., they can represent any DMP algorithm with a proper choice of learnable weights, which justifies the adoption of GNNs in wireless networks from an algorithmic perspective.

    \item To theoretically verify the advantages of GNNs over MLPs in solving graph optimizations in wireless communications, we provide the first generalization analysis for these two neural architectures. Specifically, based on the PAC-learning framework \cite{valiant1984theory,xu2019what}, we prove that the GNNs' generalization error and required number of training samples are $\mathcal{O}(n)$ and $\mathcal{O}(n^2)$ times lower than those of MLPs, where $n$ is the number of nodes in the graph. This indicates that GNNs will achieve higher performance gains in dense wireless networks.
    
    \item Based on the generalization analysis, we propose a theory-guided neural architecture design methodology for performance enhancement. Simulations will verify our theory and the effectiveness of the proposed neural architecture design principles. In particular, it will be shown that the proposed GNNs improve the performance of standard GNNs, significantly outperform MLP-based methods, and achieve hundreds of times of speedups compared with classic optimization-based methods\footnote{The codes to reproduce the simulation results are available on \url{https://github.com/yshenaw/GNN4Com}.}.
\end{enumerate}

\subsection{Preview and Organization}
In this subsection, we give a preview of the proposed framework as illustrated in Fig. \ref{fig:e2e}. We consider a general optimization problem in wireless communications \cite{sun2018learning}, with the following formulation:
\begin{align}\label{opt:gen}
    \min_{\bm{x} } \quad f(\bm{x};\bm{z}), \quad \text{s.t. } \bm{x} \in \mathcal{X},
\end{align}
where $f:\mathbb{C}^n \rightarrow \mathbb{R}$ is a continuous (possibly nonconvex) objective function, $\bm{x}$ is the optimization variable (e.g., the beamforming matrix), $\mathcal{X} \subset \mathbb{C}^n$ is the feasible set, and $\bm{z}$ contains the problem parameters (e.g., channel state information). 

Our framework consists of two aspects: \emph{algorithmic design} and \emph{theoretical analysis}. For the algorithm development, we first formulate \eqref{opt:gen} as a graph optimization problem. Specifically, we model the wireless network as a graph and the problem parameters $\bm{z}$ as the graph's features, i.e., node features $\bm{Z}$ and edge features ${\bf A}$ (Section \ref{sec:model}). We will then design a GNN architecture based on the graph type (Section \ref{subsec:mpgnn}). We further apply a theory-guided approach to modify the neural architecture for performance enhancement (Section \ref{sec:arch}). For simplicity, we adopt the data-driven paradigm\footnote{For the model-driven deep learning paradigm, one can replace the MLPs or CNNs with GNNs by following the same procedure of graph modeling and neural architecture design.}, where the GNN outputs the optimal solution $\bm{x}^*$ and is trained with an unsupervised loss function \cite{liang2018towards}. For the theoretical analysis, we first analyze the capability of GNNs in learning an optimal algorithm for Problem \eqref{opt:gen} in Theorem \ref{thm:GNNDMP} (Section \ref{sec:mpgnn}). We shall then demonstrate the benefits of GNNs over MLPs in terms of the sample complexity and generalization error in Theorem \ref{thm:gen_gnn_mlp} (Section \ref{sec:aa}). Simulation results will be presented in Section \ref{sec:exp} and finally we conclude this paper in Section \ref{sec:con}.

\begin{figure*}[htb]
	\centering
	\includegraphics[width=0.8\textwidth]{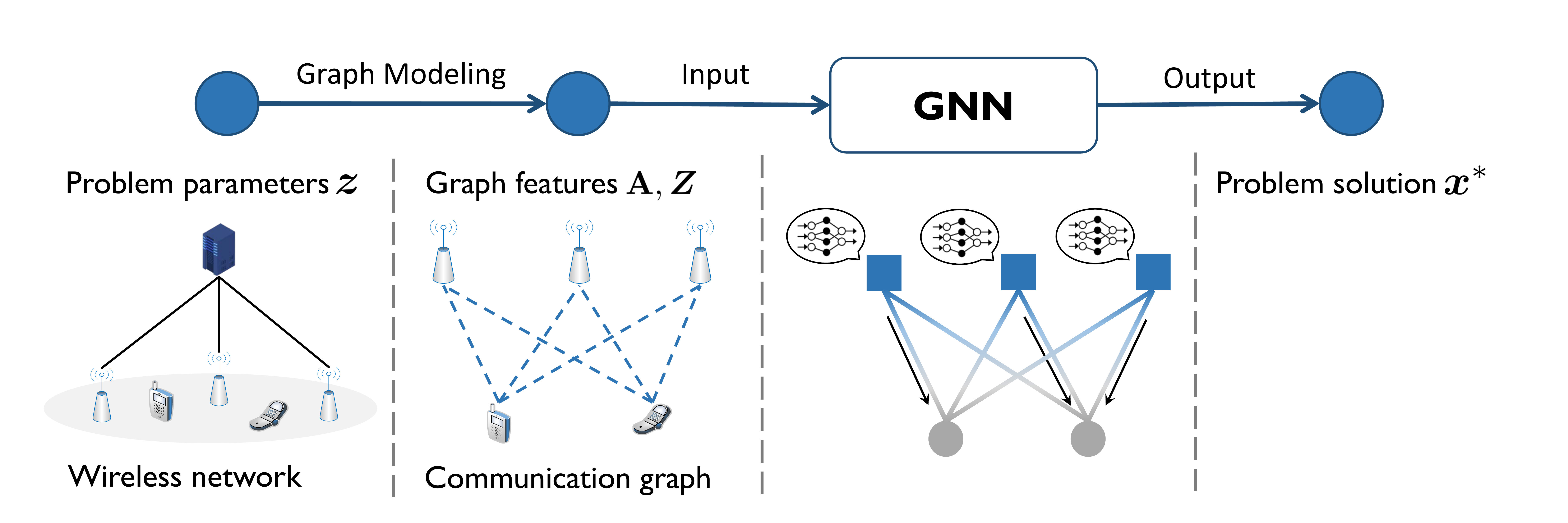}
	\caption{An illustration of the proposed algorithmic framework.}
	\label{fig:e2e}
\end{figure*}

\subsection{Notations}
Throughout this paper, $\bm{a}$, $\bm{A}$, ${\bf A}$ stand for a column vector, a matrix, and a three-dimensional tensor, respectively. $\bm{a}_{(i)}$, $\bm{A}_{(i,j)}$, and ${\bf A}_{(i,j,k)}$ are the entry of the $i$-th element of the vector $\bm{a}$, the $i$-th row and the $j$-th column of the matrix $\bm{A}$, the entry indexed by $i,j,k$ in tensor ${\bf A}$,  respectively. $\bm{A}_{(i,:)}$ denotes the $i$-th row of matrix $\bm{A}$. The conjugate, transpose, and conjugate transpose of $\bm{A}$ are presented by $\bm{A}^*$, $\bm{A}^T$, and $\bm{A}^H$; $\Re(\cdot)$ denotes the real part of a complex variable.

\begin{table*}[t]	
\selectfont  
\centering
\newcommand{\tabincell}[2]{\begin{tabular}{@{}#1@{}}#2\end{tabular}}
\caption{Typical Examples of Graph Modeling.} 
\resizebox{0.9\textwidth}{!}{
\begin{tabular}{|c|c|c|c|c|}
\hline
\textbf{Problem Type} & \textbf{System Model} & \textbf{Graph Type} & \textbf{Node/Node Features} & \textbf{Edge/Edge Features} \\ \hline
Resource management & D2D networks & \tabincell{c}{Weighted graph} &  \tabincell{c}{Transceiver pair/ \\ Direct link channels} & \tabincell{c}{Interference links/ \\ Interference link channels}   \\ \hline
Resource management & Cell-free networks & \tabincell{c}{Bipartite graph} & \tabincell{c}{UEs and APs/ \\ No feature}
& \tabincell{c}{Direct links/\\ Large-scale coefficients  } \\ \hline
Signal processing & Hybrid precoding & \tabincell{c}{Bipartite graph} & \tabincell{c}{Symbols and antennas/ \\No Feature}
& \tabincell{c}{Connectivity/ \\ Digital beamforming matrix} \\ \hline
\end{tabular}
}
\label{problem_example}
\end{table*}

\section{Graph Modeling of Wireless Communication Problems}\label{sec:model}
In this section, we first present a general graph modeling framework for wireless communication problems, followed by illustrations of three application examples that require different graph modeling techniques.
\subsection{Wireless Communication Problems as Graph Optimization}
Traditionally, a graph is an ordered pair $G = (V,E)$, where $V$ denotes the set of nodes and $E \subset \{\{x,y\}| x,y \in V\}$ denotes edges associated with pairs of nodes. In wireless networks, the agents' and links' properties are as important as the network topology. As a result, we augment the tuple $(V,E)$ to include the node features and edge features. Specifically, we define a \emph{communication graph} as $G = (V,E,s,t)$, where $V$ is the set of nodes, $E$ is the set of edges, $s: V \rightarrow \mathbb{C}^{d_1}$ maps a node to its feature, and $t: E \rightarrow \mathbb{C}^{d_2}$ maps an edge to its feature. In the context of wireless communications, the nodes represent users, access points (APs), antennas,  or variables in a factor graph, and the node feature is the property of this node (e.g., the importance of the node). The edges represent the communication links, interference links, or some connectivity patterns, while the edge feature is the property of this edge (e.g., channel state information). We further define the node feature matrix as $\bm{Z} \in \mathbb{C}^{|V| \times d_1}$ where $\bm{Z}_{(i,:) } = s(v_i)$, and the adjacency feature tensor ${\bf A} \in \mathbb{C}^{|V| \times |V| \times d_2}$ as
\begin{align*} 
{\bf A}_{ (i,j,:) } = \left\{
\begin{aligned}
&\bm{0}, &&\text{if } \{i,j\} \notin E \\
&t(\{i,j\}) &&\text{otherwise}.
\end{aligned}
\right.
\end{align*}

We now define graph optimization problems on communication graphs. Denote the optimization variable on the $i$-th node as $\bm{x}_i$ and collect all the variables as $\bm{X} = [\bm{x}_1,\cdots,\bm{x}_{|V|}]^T \in \mathbb{C}^{|V| \times n}$. A graph optimization problem is given by
\begin{equation} \label{eq:cg_opt}
\begin{aligned}
&\mathscr{P}:\underset{\bm{X}}{\text{minimize}}
& & f(\bm{X},\bm{Z},{\bf A})
& \text{ subject to }
& & Q(\bm{X},\bm{Z},{\bf A}) \leq 0
\end{aligned},
\end{equation}
where $f(\cdot)$ and $Q(\cdot)$ denote the objective function and constraint respectively. This formulation is general and includes typical examples in wireless communications, as illustrated below.

\textbf{Resource Management and Signal Processing}: First, as the topology of a wireless network is essentially a graph, radio resource management problems are naturally graph optimization problems. In this paper, we consider two example networks, namely, D2D networks and cell-free networks, where we model users and APs as nodes. Second, the graph structure is ubiquitous in signal processing applications (e.g., hybrid precoding, localization, and traffic prediction), which create graph optimization problems. We consider hybrid precoding as a signal processing example where we model the antennas as nodes. 

\textbf{Statistical Inference}: The inverse problems such as channel estimation, data detection, and wireless sensing can be formulated as optimization problems on factor graphs \cite{yedidia2000bethe}, which can be enhanced by GNNs \cite{kosasih2022graph,he2022graph}. Due to space limitation, we will not cover this topic. 

An overview of typical graph modeling examples is shown in Table \ref{problem_example}, and the following three subsections will provide more details on these three running examples.

\subsection{Power Control in Device-to-Device Networks}
\begin{figure*}
    \centering
    \subfigure[D2D networks.]
    {
        \includegraphics[width=0.62\columnwidth]{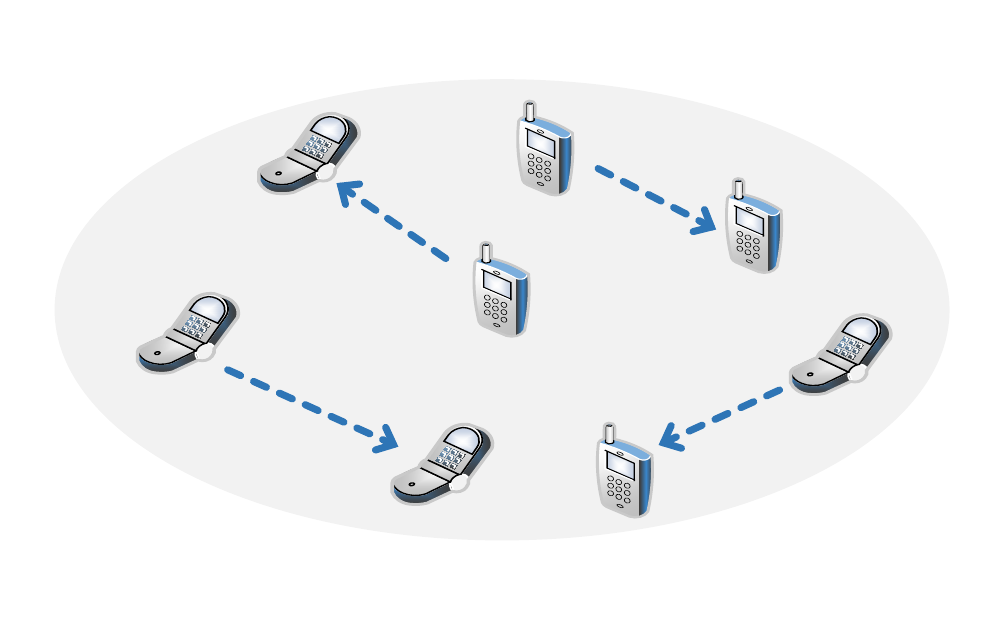}
        \label{fig:d2d_sm}
    }\hfil
    \subfigure[Cell-free networks.]
    {
        \includegraphics[width=0.62\columnwidth]{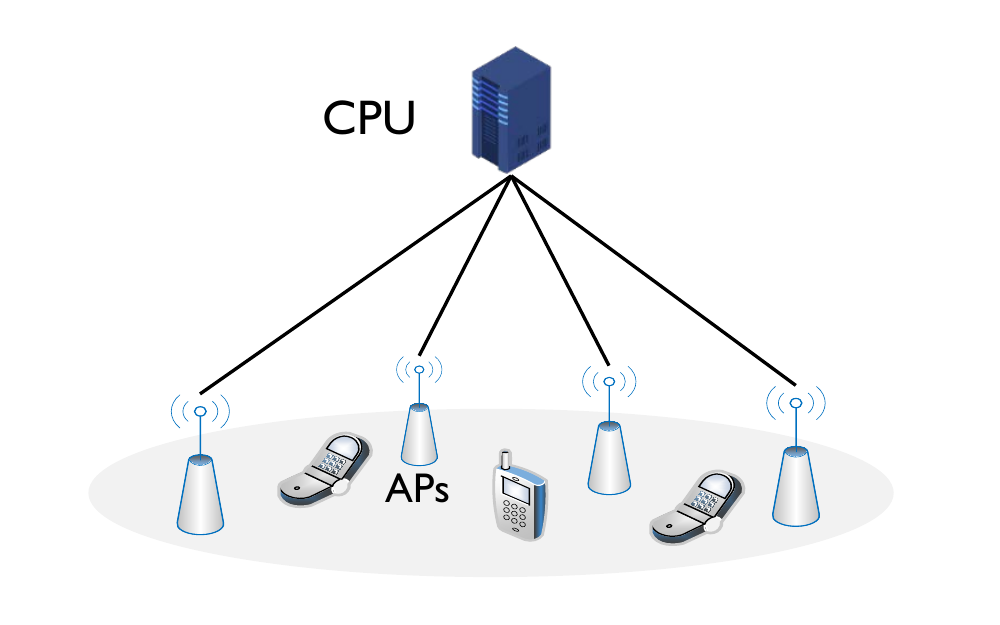}
        \label{fig:cf_sm}
    }\hfil
    \subfigure[Hybrid precoding.]
    {
        \includegraphics[width=0.62\columnwidth]{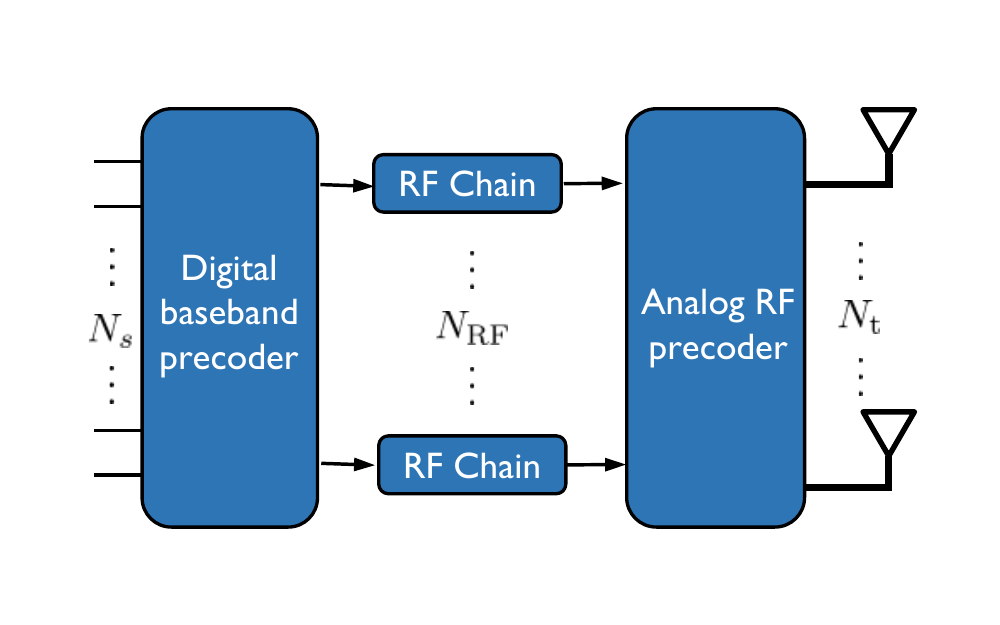}
        \label{fig:hb_sm}
    }
    \medskip
    \subfigure[Graph modeling of D2D networks.]
    {
        \includegraphics[width=0.62\columnwidth]{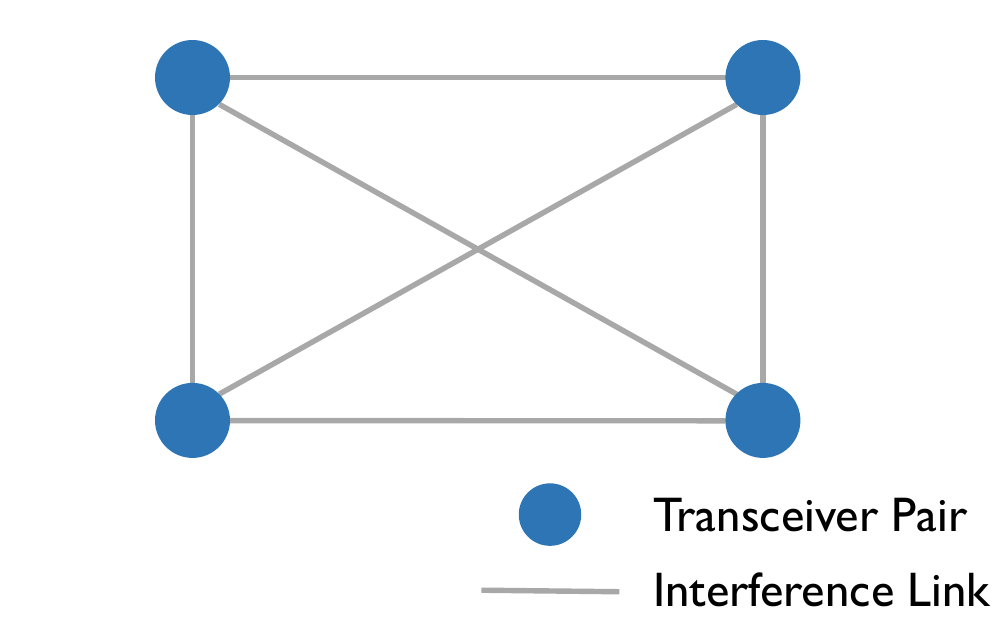}
        \label{fig:d2d_gm}
    }\hfil
    \subfigure[Graph modeling of cell-free networks.]
    {
        \includegraphics[width=0.62\columnwidth]{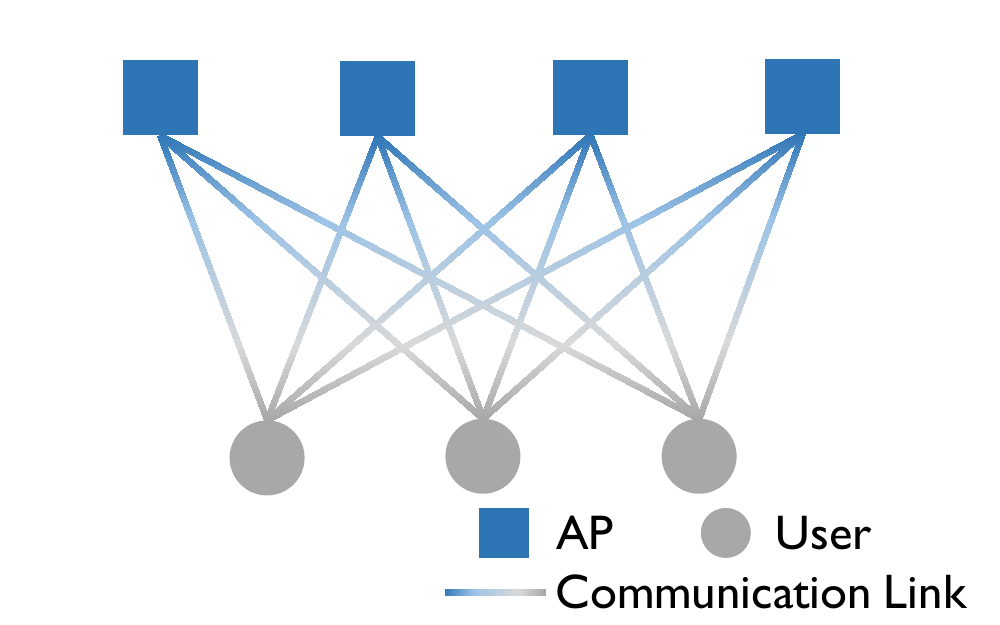}
        \label{fig:cf_gm}
    }\hfil
    \subfigure[Graph modeling of hybrid precoding.]
    {
        \includegraphics[width=0.62\columnwidth]{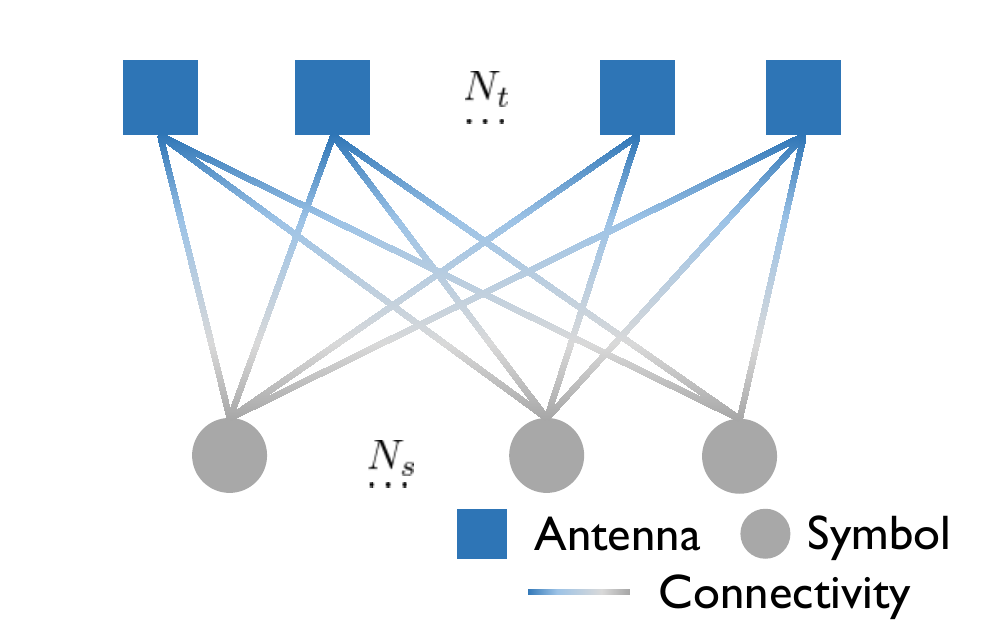}
        \label{fig:hb_gm}
    }
    \caption{Graph modeling examples.}
\end{figure*}
We first consider a D2D network with $K$ transceiver pairs where both transmitters and receivers are equipped with a single antenna \cite{sun2018learning}. Let $p_{k}$ denote the transmit power of the $k$-th transmitter, $h_{k,k} \in \mathbb{C}$ denote the direct-link channel between the $k$-th transmitter and receiver, $h_{j,k} \in \mathbb{C}$ denote the cross-link channel between transmitter $j$ and receiver $k$, $s_k \in \mathbb{C}$ denote the data symbol for the $k$-th receiver, and $n_k \sim \mathcal{CN}(0,\sigma_k^2)$ denote the additive Gaussian noise. The signal-to-interference-plus-noise ratio (SINR) for the $k$-th receiver is given by $\text{SINR}_k = \frac{|h_{k,k}|^2p_k}{\sum_{j\neq k}|h_{j,k}|^2p_j+\sigma_k^2}$. The power control problem for weighted sum rate maximization is formulated as follows, with $w_k$ representing the weight for the $k$-th pair:
\begin{equation}\label{prob:sum_rate}
\begin{aligned}
&\underset{p_1, \cdots, p_K}{\text{maximize}}
& & \sum_{k=1}^{K} w_k \log_2 \left(1+ \text{SINR}_k \right) \\
& \text{subject to}
& & 0 \leq p_k \leq 1, \forall k.
\end{aligned}
\end{equation}

To model it as a graph optimization problem, we treat the $k$-th transceiver pair as the $k$-th node in the graph. Illustrations of the system model and graph modeling are shown in Fig. \ref{fig:d2d_sm} and Fig. \ref{fig:d2d_gm}, respectively. The optimization variable $\bm{X} \in \mathbb{R}^{|V| \times 1}$, node feature matrix $\bm{Z} \in \mathbb{R}^{K \times 3}$, and adjacency feature tensor ${\bf A} \in \mathbb{R}^{K \times K \times 1}$ are given by $\bm{X}_{(k,1)} = p_k, \quad \bm{Z}_{ (k,:) } = [ |h_{k,k}|^2, w_k, \sigma_k^2]^T, \quad {\bf A}_{ (j,k,1) } = |h_{j,k}|^2$.
With the node features and edge features, the SINR is expressed as $\text{SINR}_k = \frac{\bm{Z}_{ (k,1) } \bm{X}_{(k,1)}}{ \sum_{j \neq k} {\bf A}_{ (j,k,1) }\bm{X}_{(j,1)} + \bm{Z}_{(k,3)} }$. The optimization problem \eqref{prob:sum_rate} is then converted into the following graph optimization problem: 
\begin{equation}\label{prob:sum_rate_graph}
\begin{aligned}
&\underset{\bm{X}}{\text{maximize}}
& & \sum_{k=1}^K \bm{Z}_{ (k, 2) } \log_2 \left(1+ \text{SINR}_{k} \right) \\
& \text{subject to}
& & 0 \leq \bm{X}_{(k,1)} \leq 1, \forall k.
\end{aligned}
\end{equation}

\subsection{Power Control in Cell-free Networks}
\begin{figure*}
\begin{align}\label{eq:sinr_cf}
    \text{SINR}_k = \frac{p_k \left( \sum_{m=1}^M v_{m,k}\right)^2}{\sum_{k'\neq k}^K p_{k'} \left(\sum_{m=1}^M v_{m,k}\frac{u_{m,k'}}{u_{m,k}} \right)^2 |\bm{\phi}_k^H\bm{\phi}_{k'}|^2 + \sum_{k'=1}^K p_{k'} \sum_{m=1}^M v_{m,k}u_{m,k'} + \frac{1}{\rho} \sum_{m=1}^M v_{m,k} }
\end{align}
\end{figure*}

We consider a cell-free massive MIMO system with $M$ single-antenna APs and $K$ single-antenna users, where the APs are connected to a central processing unit (CPU) \cite{ngo2017cell}. Specifically, we study the power control problem in the uplink pilot transmission phase \cite{rajapaksha2021deep}. During this stage, all the users transmit their pilot sequences of $\tau$ symbols. Let $\rho$, $\sqrt{\tau}\bm{\phi}_k \in \mathbb{C}^{\tau \times 1}$, $p_k$, $u_{m,k}$ denote the SNR, the pilot sequence and transmit power of the $k$-th user, and the large-scale fading coefficient between the $k$-th user and the $m$-th AP, respectively. The SINR of the $k$-th user is given in \eqref{eq:sinr_cf}, where $v_{m,k} = \frac{\sqrt{\tau \rho} u_{m,k}}{\tau \rho \sum_{k'=1}^K u_{m,k} |\bm{\phi}_k^H\bm{\phi}_{k'}|^2 + 1}$. We consider the following formulation with user fairness, i.e., to maximize the minimum user rate:
\begin{equation}\label{prob:cf}
\begin{aligned}
&\underset{p_k}{\text{max}} \quad \underset{k = 1, \cdots, K}{\text{min}}
& & \log_2(1+\text{SINR}_k) \\
& \text{subject to}
& & 0 \leq p_k \leq 1.
\end{aligned}
\end{equation}

For ease of comparison, we follow the setting of \cite{rajapaksha2021deep}, where only the large-scale fading coefficients $u_{m,k}$ are used as the input of neural networks. We model the system as a fully connected bipartite graph, where the APs and users are viewed as nodes. Illustrations of the system model and graph modeling are shown in Fig. \ref{fig:cf_sm} and Fig. \ref{fig:cf_gm}, respectively. The optimization variable $\bm{X} \in \mathbb{R}^{K \times 1}$ is given by 
\begin{align*}
    \bm{X}_{(k,1)} = p_k, \quad 1 \leq k \leq K.
\end{align*}
Define the large-scale coefficient matrix as $\bm{U} \in \mathbb{R}^{M \times K}$, where $\bm{U}_{(m,k)} = u_{m,k}$. The adjacency feature tensor of the bipartite graph ${\bf A} \in \mathbb{R}^{(M+K) \times (M+K) \times 1}$ is given by $ {\bf A}_{(:,:,1)} = \begin{bmatrix} \bm{0}_{M \times M} & \bm{U} \\
    \bm{U}^T & \bm{0}_{K \times K}
    \end{bmatrix}$, and there is no node feature. The graph optimization problem is obtained by replacing $p_k$ in \eqref{eq:sinr_cf} with $\bm{X}_{(k,1)}$ and replacing $u_{m,k}$ in \eqref{eq:sinr_cf} with ${\bf A}_{(m,M+k)}$. 

\subsection{Hybrid Precoding in mmWave Communications}
This example considers a single-user mmWave MIMO system,  where $N_s$ data streams are sent and collected by $N_t$ transmit antennas. Hybrid precoding is adopted to reduce hardware cost and energy consumption \cite{yu2016alternating}. The number
of RF chains at the transmitter is denoted as $N_{\rm{RF} }$, which is subject to the constraint $N_s \leq N_{\rm{RF} } \leq N_t$. The hybrid precoder consists of an $N_{\rm{RF} } \times N_s$ digital baseband precoder $\bm{F}_{\rm{BB} }$ and an $N_t \times N_{\rm{RF} }$ analog RF precoder $\bm{F}_{\rm{RF} }$. The transmit power constraint is given by $\|\bm{F}_{\rm{RF} }\bm{F}_{\rm{BB} }\|_F^2 = N_s$. We consider a generic formulation for the hybrid precoding problem \cite{yu2016alternating}:
\begin{equation}\label{prob:hybrid}
\begin{aligned}
&\underset{\bm{F}_{\rm{RF} }, \bm{F}_{\rm{BB} }}{\text{minimize}}
& & \|\bm{F}_{\rm{opt} } - \bm{F}_{\rm{RF} }\bm{F}_{\rm{BB} }\|_F \\
& \text{subject to}
& & \bm{F}_{\rm{RF} } \in \mathcal{X}, \quad \|\bm{F}_{\rm{RF} }\bm{F}_{\rm{BB} }\|_F^2 = N_s,
\end{aligned}
\end{equation}
where $\bm{F}_{\rm{opt} }$ is the optimal fully digital precoder to be approximated, and $\mathcal{X} \in \{\mathcal{X}_f, \mathcal{X}_p \}$ is the feasible set of the analog precoder induced by the unit modules constraints, which is distinct for different hybrid precoding structures. We consider the partially connected structure, where the output signal of each RF chain is only connected with $N_t/N_{RF}$ antennas \cite{yu2016alternating}. The analog precoder $\bm{F}_{\rm{RF} }$ is a block diagonal matrix and each block is an $N_t/N_{RF}$ dimensional vector with unit modulus elements:
\begin{equation}\label{eq:partial}
\bm{F}_{\rm{RF} }=\left[ {\begin{array}{*{20}{c}}
{\bm{p}_1}&{\bm{0}}&{\cdots}&{\bm{0}}\\
{\bm{0}}&{\bm{p}_2}&{}&{\bm{0}}\\
{\vdots}&{}&{\ddots}&{\vdots}\\
{\bm{0}}&{\bm{0}}&{\cdots}&{\bm{p}_{N_{\rm{RF} }}}
\end{array}} \right],
\end{equation}
where $\mathbf{p}_i=\left[\exp\left(j\theta_{(i-1)\frac{N_t}{N_\mathrm{RF}}+1}\right),\cdots,\exp\left(j\theta_{i\frac{N_t}{N_\mathrm{RF}}}\right)\right]^T$ and $\theta_i$ is the phase of the $i$-th phase shifter. 

One intuitive way of graph modeling for this example is to model the symbols, RF chains, phase shifters, and antennas as nodes in the graph. However, it is difficult to find a way to represent $\bm{F}_{\rm{opt}}$ as features on this graph. Instead, we model this system as a fully connected bipartite graph where the transmit symbols and the transmit antennas are viewed as nodes. Illustrations of the system model and graph modeling are shown in Fig. \ref{fig:hb_sm} and Fig. \ref{fig:hb_gm}, respectively. We define the optimization variables $\bm{X}_{\rm{BB}} \in \mathbb{C}^{N_s \times N_{\rm{RF}}}$ and $\bm{X}_{\rm{RF}} \in \mathbb{C}^{N_{t} \times 1}$ as
\begin{align*}
    &(\bm{X}_{\rm{BB}})_{(i,:)} = (\bm{F}_{\rm{BB}})_{(i,:)}^H, \quad (\bm{X}_{\rm{RF}})_{(i,1)} = (\bm{F}_{\rm{RF}})_{(i,\lceil  iN_t/N_{\rm{RF}}\rceil)}, \\ &1 \leq i \leq N_{\rm{RF}}, 1 \leq j \leq N_{t}.
\end{align*}
The adjacency feature tensor ${\bf A} \in \mathbb{C}^{(N_t + N_s) \times (N_t + N_s) \times 1}$ is given by
\begin{align*}
    {\bf A}_{(:,:,1)} = \begin{bmatrix} \bm{0}_{N_t \times N_t} & \bm{F}_{\rm{opt}} \\
    \bm{F}_{\rm{opt}}^H & \bm{0}_{N_s \times N_s}
    \end{bmatrix}
\end{align*}
and there is no node feature. The graph optimization problem is written as
\begin{equation}\label{prob:hybrid_partial}
\small
\begin{aligned}
&\underset{\bm{X}_{RF}, \bm{X}_{BB}}{\text{minimize}}
& & \sum_{i=1}^{N_t} \sum_{j=1}^{N_s} |{\bf A}_{(i,N_t+j,1)}- (\bm{X}_{\rm{RF} })_{(i,1)}(\bm{X}_{\rm{BB}})_{(j,\lceil iN_t/N_{\rm{RF}}\rceil) }^* |^2 \\
& \text{subject to}
& & |(\bm{X}_{\rm{RF} })_{(i,1)}| = 1, \quad 1 \leq i \leq N_t, \\
& 
& & \|\bm{X}_{\rm{BB} }\|_F^2 = \frac{N_{\rm{RF}}N_s}{N_t}.
\end{aligned}
\end{equation}
\section{From Wireless Networks to Graph Neural Networks}\label{sec:mpgnn}
In this section, we first introduce GNNs, and illustrate how to apply them to solve the graph optimization problem formulated in \eqref{eq:cg_opt}. Then, we develop a unified framework that bridges optimization-based algorithms and GNN-based algorithms, which justifies the application of GNNs to solve problems in wireless communications from an algorithmic perspective.

\subsection{Graph Neural Networks}\label{subsec:mpgnn}
GNNs extend the spatial convolution in CNNs to graphs. In one CNN layer, each pixel updates its own hidden state based on the convolution with its neighbor pixels. Similarly, in one GNN layer, each node updates its own hidden state based on the aggregated information from its neighbor nodes. Specifically, denote $\bm{d}^{(t)}_k$ as the hidden state of the $k$-th node at the $t$-th layer, each GNN layer consists of two stages \cite{wang2019deep,shen2020graph,guo2021learning}:
\begin{enumerate}
    \item \textbf{Aggregation}: The $k$-th node uses a neural network to aggregate its neighbors' outputs of the last layer, followed by a pooling function. The update can be expressed as:
    \begin{align}\label{gnn:agg}
        \bm{a}_k^{(t)} = \text{PL}_{j \in \mathcal{N}(k)} \left(q_1^{(t)}(\bm{d}_k^{(t-1)}, \bm{d}_j^{(t-1)}, {\bf A}_{(j,k,:)}|\mathcal{W}_1^{(t)} )\right),
    \end{align}
    where $\mathcal{N}(k)$ is the set of neighbors of node $k$, $\text{PL}_{j \in \mathcal{N}(k)}(\cdot)$ is a pooling function (e.g., sum or max pooling), $q_1^{(t)}(\cdot)$ is the aggregation neural network at the $t$-th layer, and $\mathcal{W}_1^{(t)}$ denotes the learnable weights of this neural network. 
    \item \textbf{Combination}: After obtaining the aggregated information, another neural network is applied to process information and obtain the output at each node. Specifically, the aggregated information is combined with the node's own information as follows
    \begin{align}\label{gnn:comb}
        \bm{d}_k^{(t)} = q_2^{(t)}(\bm{d}_k^{(t-1)}, \bm{Z}_{(k,:)},\bm{a}_k^{(t)}|\mathcal{W}_2^{(t)}),
    \end{align}
    where $q_2^{(t)}(\cdot)$ represents the combination neural network and is with learnable model parameters $\mathcal{W}_2^{(t)}$.
\end{enumerate}

To apply a GNN on a particular graph, one should design $q_1$, $q_2$, and the pooling function. Recently the machine learning community has developed various kinds of GNN architectures for different types of graphs, and a straightforward way is to select the one that works well on the corresponding graph. For example, D2D networks are weighted graphs, and thus one may select a GNN that is designed for weighted graphs from the deep graph library (DGL) \cite{wang2019deep}, as illustrated below.

\begin{exmp}\label{exmp:ecgnn} (ECGNN for D2D networks) The edge convolutional graph neural network (ECGNN) is an architecture designed for weighted graphs and has been applied to D2D power control \cite{shen2019graph}. It first initializes $\bm{d}^{(0)}_k = \bm{Z}_{(k,:)} = [w_k, h_{k,k}]$. For a $T$-layer ECGNN, the update of the $k$-th node in the $t$-th layer is specified as
\begin{equation}\label{gnn:ecgnn}
\begin{aligned}
    &\bm{a}_k^{(t)} = \sum_{j \neq k} \sigma(\bm{W}_1^{(t)} \bm{d}^{(t-1)}_k + \bm{w}_2^{(t)} h_{j,k} + \bm{W}_3^{(t)} \bm{d}^{(t-1)}_j), \\
    &\bm{d}_k^{(t)} = \frac{\bm{a}_k^{(t)}}{|\mathcal{N}(k)|} , \quad p_k = \text{MLP}(\bm{d}_k^{(T)}), 
\end{aligned}
\end{equation}
where $\bm{W}_1,\bm{W}_2, \bm{W}_3$ are the weight matrices, $\text{MLP}(\cdot)$ is a learnable MLP mapping the hidden state to the power values, and $\sigma(\cdot)$ is an activation function. In this GNN, we have $\text{PL} = \text{SUM},  \mathcal{W}_1^{(t)} = \{\bm{W}_1^{(t)}, \bm{w}_2^{(t)}, \bm{W}_3^{(t)}\}, \mathcal{W}_2^{(t)} = \emptyset$.
\end{exmp}

Different from D2D networks, the cell-free network is modeled as a bipartite graph. As a result, we may choose a neural network designed for bipartite graphs from DGL \cite{wang2019deep}.

\begin{exmp}\label{exmp:hetgnn} (HetGNN for cell-free networks) The heterogeneous graph neural network (HetGNN) \cite{guo2021learning} is powerful in learning representations of heterogeneous bipartite graphs. In each layer of a HetGNN, there are two types of messages: the messages from APs to UEs, and the messages from UEs to APs. As a result, we should use different weight matrices to parameterize different messages. As there is no node feature, we initialize $\bm{d}^{(0)}_k \in \mathbb{R}^0$ as an empty vector. For a $T$-layer HetGNN, the update of the $k$-th node in the $t$-th layer is given as
\begin{equation}\label{gnn:hetgnn}
\begin{aligned}
    &\bm{a}^{(t)}_k = \left\{
\begin{aligned}
&\sum_{m=1}^M \bm{W}_1^{(t)}\bm{d}^{(t-1)}_m + \bm{w}_2^{(t)} u_{m,k}, \quad  M < k \leq M+K, \\
&\sum_{m=M+1}^{M+K} \bm{W}_3^{(t)}\bm{d}^{(t-1)}_m + \bm{w}_4^{(t)} u_{k,m}, \quad 0 \leq k \leq M,
\end{aligned}
\right. \\\
    &\bm{d}^{(t)}_k = \left\{
\begin{aligned}
&\sigma(\bm{W}_5^{(t)}\bm{d}^{(t-1)}_k + \bm{a}^{(t)}_k), \quad  M < k \leq M+K, \\
&\sigma(\bm{W}_6^{(t)}\bm{d}^{(t-1)}_k + \bm{a}^{(t)}_k), \quad 0 \leq k \leq M,
\end{aligned}
\right. \\
&p_k = \text{MLP}(\bm{d}^{(T)}_k), \quad  M < k \leq M + K,
\end{aligned}
\end{equation}
where $\bm{W}_1^{(t)}, \cdots, \bm{W}_6^{(t)}$ are learnable weights, $\sigma(\cdot)$ is an activation function, and $\text{MLP}(\cdot)$ is a learnable MLP mapping the hidden state to the power values. In this GNN, we have $\text{PL} = \text{SUM}, \mathcal{W}^{(t)} = \{\bm{W}_1^{(t)}, \bm{W}_3^{(t)}, \bm{w}_2^{(t)}, \bm{w}_4^{(t)}\},  \mathcal{W}_2^{(t)} = \{\bm{W}_5^{(t)}, \bm{W}_6^{(t)}\}$.
\end{exmp}

As the underlying topology of hybrid precoding is also a bipartite graph, HetGNN can be applied. Nevertheless, a good architecture design for this problem is more involved and will be treated in Section \ref{arch:hb}.

\subsection{Distributed Message Passing Algorithms}
In this subsection, we introduce \emph{distributed message passing} (DMP) algorithms \cite{angluin1980local}, which play an essential role in the analysis of GNNs.  They are also named after \emph{LOCAL algorithms} in distributed computing literature \cite{loukas2019graph}. DMP algorithms are a classic family of iterative optimization algorithms on graphs. In each iteration, each node sends messages to its neighbors, receives messages from its neighbors, and updates its state based on the received messages. The updates are shown in Algorithm \ref{alg:dmp}, where $h^{(t)}$, $g_1^{(t)}$, $g_2^{(t)}$ are the message encoding function, message aggregation function, and update function at the $t$-th iteration, respectively. The choice of these functions depends on the specific application. This class of algorithms is extremely powerful and is able to solve any graph optimization problem, as shown in the following theorem.

\begin{algorithm}\label{alg:dmp}
	\caption{Distributed Message Passing Algorithms \cite{angluin1980local}}  
	\label{alg:dmp}  
	\small
	\begin{algorithmic}[1]
		\State Initialize all internal states $\bm{x}_i^{(0)}, \forall i$.
		\For{communication round $t=1,\cdots,T$}
		\State $\forall i$, agent $i$ receives 
        \Statex $\left\{\bm{m}_{j\rightarrow i}^{(t)}|\bm{m}_{j\rightarrow i}^{(t)} = h^{(t)}\left(\bm{x}_j^{(t-1)}, \bm{x}^{(t-1)}_i,  {\bf A}_{(j,i,:)}\right), j\in \mathcal{N}(i)\right\}$ from the edges.
		\State $\forall i$, agent $i$ aggregates messages $\bm{y}^{(t)}_i = g_1^{(t)}\left(\left\{\bm{m}^{(t)}_{j\rightarrow i}: j \in \mathcal{N}(i)\right\}\right)$.
		\State $\forall i$, agent $i$ updates its internal state  $\bm{x}^{(t)}_i = g_2^{(t)}\left(\bm{Z}_{(i,:)}, \bm{x}^{(t-1)}_i, \bm{y}_i^{(t)}\right)$.
		\EndFor
		\State Output $[\bm{x}_i^{(T)}]_{i=1}^{|V|}$.
	\end{algorithmic}  
\end{algorithm}

\begin{thm}\label{fact:dmp}
    For any graph optimization problem as formulated in \eqref{eq:cg_opt}, there exists a DMP algorithm that can solve it.
\end{thm}

\begin{proof} The proof of Theorem \ref{fact:dmp} relies on two facts. First, as shown in \cite{shen2020graph}, any graph optimization problem is permutation invariant. Second, any permutation invariant function can be implemented by a DMP step. Thus, for any graph optimization problem, there exists a DMP algorithm that can solve it. Details are provided in Appendix \ref{app:dmp}.
\end{proof}

Although the name of ``distributed message passing algorithms'' seldom appears in the wireless communication literature, many classic algorithms belong to the DMP algorithm class. The following results show that classic algorithms for hybrid precoding and power control are DMP algorithms, which are proved by writing the iterations of classic algorithms as DMP steps. 

\begin{prop} (Riemannian gradient as a DMP algorithm) The Riemannian gradient algorithm in  \cite{yu2016alternating} for Problem \eqref{prob:hybrid} is a DMP algorithm. 
\end{prop}

\begin{proof}
We first write down the key updates of the Riemannian gradient:
\begin{equation}
\small
    \begin{aligned}
    &\nabla f(\bm{F}_{\rm{RF}}^{(t-1)}) = -2(\bm{F}_{\rm{opt}} - \bm{F}_{\rm{RF}}^{(t-1)} \bm{F}_{\rm{BB}}^{(t-1)}) (\bm{F}_{\rm{BB}}^{(t-1)})^H,\\
    &\mathrm{grad}f(\bm{F}_{\rm{RF}}^{(t-1)}) = \nabla f(\bm{F}_{\rm{RF}}^{(t)}) - \Re\{\nabla f(\bm{F}_{\rm{RF}}^{(t)}) \odot (\bm{F}_{\rm{RF}}^{(t-1)})^* \} \odot \bm{F}_{\rm{RF}}^{(t-1)},\\
    &\bm{F}_{\rm{RF} }^{(t)} = P_{\mathcal{X}}\left(\bm{F}_{\rm{RF} }^{(t-1)} - \alpha^{(t)} \cdot \mathrm{grad}f(\bm{F}_{\rm{RF}}^{(t)})\right).\label{eq:RG}
\end{aligned}
\end{equation}
Here, $\alpha^{(t)}$ is the step size at the $t$-th iteration and $P_{\mathcal{X}}$ is the projection onto the unit modulus constraint. We next show that \eqref{eq:RG} can be expressed as a DMP step, where the messages are passed from the transmit symbol nodes to the transmit antenna nodes.  The hidden state of the $i$-th node $\bm{x}_i^{(t)} \in \mathbb{C}^{N_{\rm{RF}}}$ is
\begin{align}\label{eq:edge_fea}
\bm{x}_i^{(t)} = \left\{
\begin{aligned}
&(\bm{X}_{\rm{RF}}^{(t)})_{(i,:)}^T, && 1 \leq i \leq N_t \\
&(\bm{X}_{\rm{BB}}^{(t)})_{(i-N_t,:)}^T && N_t < i,
\end{aligned}
\right. 
\end{align}
where $\bm{X}_{\rm{RF}}^{(t)} = \bm{F}_{\rm{RF}}^{(t)}$ and $\bm{X}_{\rm{BB}}^{(t)} = (\bm{F}_{\rm{BB}}^{(t)})^H$. With notations in Section \ref{sec:model}, the corresponding DMP is 
\begin{align*}
    &\bm{m}_{j \rightarrow k}^{(t)} = ({\bf A}_{(j,k,1)}^* - (\bm{x}_k^{(t-1)})^T (\bm{x}_j^{(t-1)})^*)\bm{x}_j^{(t-1)}, \\
    & 1 \leq k \leq N_t,  N_t < j \leq N_t + N_s, \\
    &\bm{y}_{k}^{(t)} = - 2\sum_{j=N_t+1}^{N_t+N_s} \bm{m}_{j\rightarrow k}^{(t)}, \quad 1 \leq k \leq N_t,\\
    &\hat{\bm{x}}_k^{(t)} = \bm{y}_{k}^{(t)} - \Re\{\bm{y}_{k}^{(t)} \odot (\bm{x}_k^{(t-1)})^* \} \odot (\bm{x}_k^{(t-1)}), \quad  1 \leq k \leq N_t,\\
    &\bm{x}_k^{(t)} = P_{\mathcal{X}}\left(\bm{x}_k^{(t-1)} - \alpha \cdot \hat{\bm{x}}_k^{(t)}\right), \quad  1 \leq k \leq N_t,
\end{align*}
where $\bm{m}_{j \rightarrow k}^{(t)} \in \mathbb{C}^{N_{\rm{RF}}}$ is the message sent from node $j$ to node $k$, $\bm{y}_{k}^{(t)} \in \mathbb{C}^{N_{\rm{RF}}}$ is the aggregated messages at the node $k$, and $\hat{\bm{x}}_k^{(t)} \in \mathbb{C}^{N_{\rm{RF}}}$ is an intermediate variable. 
\end{proof}

\begin{prop} (WMMSE as a DMP algorithm) The weighted minimum mean square error (WMMSE) algorithm (Algorithm 1 in \cite{sun2018learning}) for solving Problem \eqref{prob:sum_rate} is a DMP algorithm. 
\end{prop}

\begin{proof}
We first write down the key updates of WMMSE \cite{sun2018learning}
\begin{subequations}\label{eq:wmmse}
\begin{align}
    &v_k^{(t)} = \frac{w_k \alpha_k^{(t-1)} u_k^{(t-1)} |h_{k,k}| }{\sum_{j=1}^K w_j \alpha_j^{(t-1)} (u_j^{(t-1)})^2 |h_{j,k}|^2 }, \label{eq:wmmse_1}\\
    &u_k^{(t)} = \frac{|h_{k,k}|v_k^{(t)}}{\sum_{j=1}^K |h_{k,j}|^2(v_j^{(t)})^2 + \sigma_k^2 }, \label{eq:wmmse_2}\\
    &\alpha_k^{(t)} = \frac{1}{1 - u_k^{(t)} |h_{k,k}|v_k^{(t)} } \label{eq:wmmse_3},
\end{align}
\end{subequations}
where the output $p_k = (v_k)^2$ and $v_k^t, u_k^t, \alpha_k^t$ are the intermediate variables at the $t$-th iteration. We show that one iteration of \eqref{eq:wmmse_1} can be written as one iteration of a DMP step. Specifically, denoting the hidden state $\bm{x}_i^{(t)} = [v_k^{(t)},u_k^{(t)},\alpha_k^{(t)}] \in \mathbb{R}^3$ and with the notations in Section \ref{sec:model}, we have
\begin{align*}
    &m_{j\rightarrow k}^{(t)} = \bm{Z}_{(j,2)} \cdot (\bm{x}^{(t-1)}_j)_{(3)} \cdot ( (\bm{x}^{(t-1)}_j)_{(2)})^2 \cdot {\bf A}_{(j,k,1)}, \quad j \neq k, \\
    &y_k^{(t)} = \sum_{j=1}^K m_{j\rightarrow k}^{(t)}, \\
    &\bm{x}^{(t)}_k =  \left[\frac{\bm{Z}_{(k,2)} \cdot (\bm{x}^{(t-1)}_k)_{(2)} \cdot (\bm{x}^{(t-1)}_k)_{(1)} \cdot \sqrt{\bm{Z}_{(k,1)}} }{y_k^{(t)} + \bm{Z}_{(k,2)} \cdot (\bm{x}^{(t-1)}_k)_{(3)} \cdot ( (\bm{x}^{(t-1)}_k)_{(2)})^2 \cdot \bm{Z}_{(k,1)} },\right. \\
    &\left. (\bm{x}^{(t-1)}_k)_{(2)}, (\bm{x}^{(t-1)}_k)_{(3)}\right],
\end{align*}
where $m_{j \rightarrow k}^{(t)} \in \mathbb{R}$ and $y_k^{(t)} \in \mathbb{R}$ represent the message and aggregated message respectively. Similarly, \eqref{eq:wmmse_2} and \eqref{eq:wmmse_3} together are one iteration of a DMP step. Thus, a WMMSE algorithm with $T$ iterations is a DMP algorithm with at most $2T$ iterations.
\end{proof}

\textbf{Connections to message passing algorithms on graphical models:} In addition to these optimization algorithms, another important class of DMP algorithms is belief propagation (BP) type algorithms, e.g., belief propagation, generalized expectation propagation, and approximate message passing. The reason is that BP-type algorithms can be viewed as optimizers for free energy minimization on factor graphs \cite{yedidia2000bethe}. This is essentially a graph optimization problem, and thus Theorem \ref{fact:dmp} applies.

\subsection{Unifying Optimization-based and GNN-based Algorithms}
In this subsection, we bridge GNNs and optimization-based algorithms via DMP and demonstrate the capacity of GNNs in learning optimal algorithms. We begin by noting that GNNs are special cases of DMP algorithms. For example, let $\bm{x}_k^{(t)} = \bm{d}_k^{(t)}$, then ECGNN in \eqref{gnn:ecgnn} can be transformed into the following DMP:
\begin{align*}
    &\bm{m}_{j \rightarrow k}^{(t)}  = \sigma(\bm{W}_1^{(t)} \bm{x}^{(t-1)}_k + \bm{w}_2 {\bf A}_{(j,k,1)} + \bm{W}_3\bm{x}^{(t-1)}_j), \\
    &\bm{y}_k^{(t)} = \frac{1}{\mathcal{N}(k)}\sum_{j} \bm{m}_{j\rightarrow k}^{(t)}, \quad
    \bm{x}_k^{(t)} = \bm{y}_k^{(t)}.
\end{align*}
Also, due to the strong approximation ability of neural networks \cite{hornik1989multilayer}, they are able to represent DMP algorithms well. A formal statement of the relationship between GNNs and DMP algorithms is given in the following theorem.

\begin{thm}\label{thm:GNNDMP} 
	Let DMP($T$) denote the family of DMP algorithms with $T$ iterations and GNN($T$) as the family of GNNs in \eqref{gnn:agg} and \eqref{gnn:comb} with $T$ layers, then 
	\begin{enumerate}
		\item for any GNN in GNN($T$), there exists a distributed message passing algorithm in DMP($T$) that solves the same set of problems;
		\item for any algorithm in DMP($T$), there exists a GNN in GNN($T$) that solves the same set of problems.
	\end{enumerate}
\end{thm}

\begin{proof}
This result is proved by following Theorem 7 in \cite{shen2020graph} or Theorem 3.1 in \cite{loukas2019graph} and adding $\bm{x}_i^{(t)}$ into the messages.
\end{proof}

Theorem \ref{thm:GNNDMP} shows that GNNs are special cases of DMP algorithms, but they can represent any DMP algorithm with proper weights that are learned from data. For classic optimization-based DMP algorithms, the message encoding function, message aggregation function, and update function are handcrafted. On the contrary, GNNs directly learn these functions from data, and thus the number of iterations can be significantly reduced and the performance may be improved. For example, for the power control problem, a GNN with $2$ iterations already outperforms WMMSE, a handcrafted DMP algorithm, with $30$ iterations (Fig. 3 in \cite{shen2020graph}). Such advantage is a major motivation underlying the recent interests in ``learning to optimize'' approaches \cite{sun2018learning,shen2018lora,shen2020graph}. By combining Theorem \ref{fact:dmp} and Theorem \ref{thm:GNNDMP}, we reach the conclusion that GNNs are powerful in solving graph optimization problems, which is stated in the next corollary.

\begin{cor}\label{cor:exist}
    For any graph optimization problem as defined in \eqref{eq:cg_opt}, there exists a GNN that can solve it.
\end{cor}

Corollary \ref{cor:exist} demonstrates the potential of GNNs in solving graph optimization problems. But this support from the algorithmic perspective alone cannot fully explain the advantages of GNN-based methods. Particularly, as MLPs are also universal approximators \cite{hornik1989multilayer,sun2018learning}, they have the ability to approximate DMP algorithms too. Thus, we still need a result to directly compare GNN-based and MLP-based methods. This is left to the next section, where we will leverage a fine-grained analysis on the generalization performance to characterize the benefits of adopting GNNs over MLPs.

\section{Generalization Analysis of Graph Neural Networks}\label{sec:aa} 
In this section, we theoretically prove that GNNs are superior to MLPs in solving problems in wireless networks, in both the generalization error (i.e., the performance of solving a problem instance of \eqref{eq:cg_opt}) and sample efficiency (i.e., the number of samples to train the neural network). We first introduce PAC-learning as a general analytical framework and present the algorithmic alignment (AA) as the main tool. We will then characterize and compare the generalization error and sample efficiency of GNNs and MLPs in solving graph optimization problems.
\subsection{Generalization Analysis via PAC-learning}
PAC-learning is a standard framework to analyze the generalization error of machine learning methods with a finite amount of training samples. Intuitively, it addresses the following problem: \emph{How many training samples are needed to guarantee that the generalization error is less than $\epsilon$ with a probability of at least $1-\delta$}. This framework is formally defined as follows.

\begin{defn}\label{def:PAC} (PAC Learning \cite{valiant1984theory,xu2019what}) Fix an error parameter $\epsilon > 0$ and a failure probability $\delta \in (0,1)$. Suppose $\{\bm{x}_i,\bm{y}_i\}_{i=1}^N$ are $N$ i.i.d. samples drawn from a distribution $\mathcal{D}$ with each sample satisfying $y_i = g(\bm{x}_i)$. Let $f=\mathcal{A}(\{\bm{x}_i,\bm{y}_i\}_{i=1}^N)$ be the function generated by a learning algorithm $\mathcal{A}$. Then $g$ is $(N, \epsilon, \delta)$-learnable with algorithm $\mathcal{A}$ if $\mathbb{P}_{\bm{x} \sim \mathcal{D}} [\|f(\bm{x}) - g(\bm{x})\| \leq \epsilon ] \geq 1 - \delta$. The sample complexity, denoted by $C_{\mathcal{A}}(g,\epsilon,\delta)$, is the minimum $N$ that guarantees $g$ is $(N,\epsilon,\delta)$-learnable with $\mathcal{A}$.
\end{defn}

When applying deep learning to solve optimization problems in wireless communications, $g(\cdot)$ is an oracle algorithm that can output an optimal solution for the considered problem (e.g., a DMP algorithm), $f(\cdot)$ is a neural network, and $\mathcal{A}$ is the gradient descent algorithm to train the neural network. Intuitively, we want to learn an algorithm $f$ with a low sample complexity (i.e., requiring a small number of samples), so that $f$ is close to the oracle algorithm $g$ with a high probability. Applying the PAC learning framework for generalization analysis requires the existence of an $f$ that can well approximate $g$. The universal approximation property of neural networks provides this guarantee \cite{hornik1989multilayer}. The next step is to characterize the sample complexity for a given neural architecture, i.e., how hard it is to learn $f$ with the gradient descent algorithm. This  will be provided in the next subsection.

\subsection{Analysis of Neural Networks via Algorithmic Alignment}
Our main purpose in the generalization analysis is to compare the effectiveness of different neural architectures in solving graph optimization problems. As far as we know, the recently proposed algorithmic alignment (AA) framework \cite{xu2019what} is the only approach that allows tractable comparisons between deep neural architectures for PAC learning. The AA framework is formally defined as follows.

\begin{defn}\label{def:aa} (Algorithmic alignment \cite{xu2019what}) Let $g$ be an oracle function that generates labels, and $\mathcal{M}$ is a neural network consisting of sub neural networks $\mathcal{M}_i, i = 1,\cdots,n$. The module functions $f_1, \cdots, f_n$ generate $g$ for $\mathcal{M}$ if, by replacing $\mathcal{M}_i$ with $f_i$, the neural network $\mathcal{M}$ is identical to $g$. The neural network $\mathcal{M}$ is said to $(N,\epsilon,\delta)$-algorithmically align with the oracle function $g$ if (1) $f_1,\cdots, f_n$ generate $g$; (2) there are learning algorithms $\mathcal{A}_i$ for the $\mathcal{M}_i$'s such that $n\cdot \max_{i} \mathcal{C}_{\mathcal{A}_i}(f_i,\epsilon,\delta) \leq N$.
\end{defn}

In Definition \ref{def:aa}, the original task is to learn $g$, which is a challenging problem. If we could find a neural network architecture that is suitable for this task, it helps to break the original task into simpler sub-tasks, i.e., to learn $f_1, \cdots, f_n$ instead. The first condition, i.e., $f_1,\cdots, f_n$ generate $g$, is needed to guarantee the optimal substructure. This means that if the sub-tasks $f_1, \cdots, f_n$ can be learned optimally, then $g$ can be learned optimally. The second condition is to provide the sample complexity guarantee, which is defined in the PAC-learning framework. Therefore, once we establish the algorithmic alignment between an algorithm and a neural network, we can analyze the sample complexity and generalization error, as revealed by the next theorem. 

\begin{thm}\label{thm:aa} (PAC-learning via AA \cite{xu2019what}) Fix parameters $\epsilon$ and $\delta$. Suppose training samples are generated as $\{S_i,y_i\}_{i=1}^M \sim \mathcal{D}$ with an oracle function $y_i = g(S_i)$ for some $g$. Suppose $\mathcal{M}_1,\cdots, \mathcal{M}_n$ are the neural network modules in the sequential order, and the neural network $\mathcal{M}$ composed by $\mathcal{M}_1,\cdots, \mathcal{M}_n$ and the oracle function $g$ are $(N,\epsilon,\delta)$-algorithmically aligned via functions $f_1,\cdots, f_n$. Then, $g$ is $(N,O(\epsilon),O(\delta))$-learnable by $\mathcal{M}$.
\end{thm}

In Definition \ref{def:aa}, the degree of alignment is measured by the sample complexity, i.e., a smaller $N$ implies better alignment. Theorem \ref{thm:aa} further demonstrates that a better alignment results in a smaller sample complexity and generalization error. Thus, the AA framework can provide direct guidelines for designing neural architectures. To apply this powerful framework to comparing GNNs and MLPs, we first establish the algorithmic alignment between GNNs and DMP algorithms. The following example provides an illustration for a particular but representative GNN architecture.

\begin{exmp}\label{exmp:aa} (Alignment between GNNs and DMP algorithms) In this example, we construct a representative GNN, which algorithmically aligns with DMP algorithms. For the GNNs introduced in Section \ref{subsec:mpgnn}, the design space is the aggregation neural network $q_1(\cdot)$, combination neural network $q_2(\cdot)$, and the pooling function. We use MLPs as $q_1(\cdot)$ and $q_2(\cdot)$, which are commonly adopted in most GNNs \cite{wang2019deep}. The pooling function has various choices, e.g., sum, max, and attention. To be general, we do not explicitly specify a pooling function and instead let $q_2(\cdot)$ learn a suitable pooling function\footnote{For other pooling functions, we can follow the procedure in this example and obtain a bound of the same order.}. Specifically, we consider the following GNN:
\begin{equation}\label{gnn:wcgcn}
    \begin{aligned}
    \bm{a}_k^{(t)} &= \text{CONCAT}(\text{MLP1}^{(t)}(\bm{d}_j^{(t-1)}, \bm{d}_k^{(t-1)}, {\bf A}_{(j,k,:)})), \\
    \bm{d}_k^{(t)} &= \text{MLP2}^{(t)}(\bm{Z}_{(k,:)}, \bm{a}_k^{(t)}, \bm{x}_k^{(t-1)}),
\end{aligned}
\end{equation}
where $\text{MLP1}^{(t)}(\cdot)$ and $\text{MLP2}^{(t)}(\cdot)$ are two learnable MLPs at the $t$-th layer, and $\text{CONCAT}(\cdot)$ concatenates all the input into a vector. To apply the AA framework, we define the module functions $f_1, \cdots, f_{2T}$ and neural network modules $\mathcal{M}_1, \cdots, \mathcal{M}_{2T}$ as
\begin{equation}\label{eq:align}
    \begin{aligned}
    &f_i = \left\{
\begin{aligned}
&h^{(\lfloor i/2 \rfloor)}, &&\text{if }i \text{ is odd},  \\
&g_2^{(\lfloor i/2 \rfloor)} \circ g_1^{(\lfloor i/2 \rfloor)}, && \text{o.w.},
\end{aligned}
\right. \\ 
&\mathcal{M}_i = \left\{
\begin{aligned}
&\text{MLP1}^{(\lfloor i/2 \rfloor)}, &&\text{if }i \text{ is odd},  \\
&\text{MLP2}^{(\lfloor i/2 \rfloor)}, && \text{o.w.},
\end{aligned}
\right.
\end{aligned}
\end{equation}
where $h$, $g_1$, and $g_2$ are the functions in Algorithm \ref{alg:dmp}, and $\circ$ is a composition function. By replacing $\mathcal{M}_i$ with $f_i$ in \eqref{gnn:wcgcn}, we obtain the following algorithm:
\begin{align*}
    &\bm{a}_k^{(t)} = \text{CONCAT}(h^{(t)}(\bm{d}_j^{(t-1)}, \bm{d}_k^{(t-1)}, {\bf A}_{(j,i,:)})), \\
    &\bm{d}_k^{(t)} = g_2^{(t)} (\bm{Z}_{(k,:)}, \bm{x}_k^{(t-1)}, g_1^{(t)} (\bm{a}_k^{(t)}) ),
\end{align*}
which is identical to Algorithm \ref{alg:dmp}. In this way, the first condition in Definition \ref{def:aa} is satisfied. The second condition is satisfied if we set
\begin{equation}\label{bound:C}
\small
    \begin{aligned}
    N > 2T \cdot \frac{C}{|V|}, \quad C = \max_{t = 1, \cdots, T} (C_{\mathcal{A}}(h^{(t)},\epsilon, \delta), C_{\mathcal{A}}(g_1^{(t)} \circ g_2^{(t)},\epsilon, \delta)),
\end{aligned}
\end{equation}
where $C_{\mathcal{A}}(\cdot)$ is defined in Definition \ref{def:PAC}. This argument is more complicated and the proof is given in Appendix \ref{app:lem_gnn_mlp}. 
\end{exmp}

Next, we will compare the generalization performance of GNN-based and MLP-based methods. To apply MLP-based methods for the graph optimization problems in \eqref{eq:cg_opt}, we first flatten the node features $\bm{Z}$ and edge features ${\bf A}$ into a vector \cite{sun2018learning,liang2018towards}. Then an MLP is adopted to map the input vector to the solution of Problem \eqref{eq:cg_opt}. Intuitively, MLPs are unstructured, and thus they will not align with the DMP algorithm for solving \eqref{eq:cg_opt} as well as GNNs. The following result rigorously verifies this conclusion.

\begin{lem}\label{lem:gnn_mlp} (GNNs aligns with DMP algorithms better than MLPs)
    Suppose the target function is a $T$-iteration DMP algorithm $D$, then the GNN architecture in \eqref{gnn:wcgcn}  $(O(TC/|V|), O(\epsilon), O(\delta))$-algorithmically aligns with $D$ while the MLP architecture $(O(TC|V|), O(\epsilon), O(\delta))$-algorithmically aligns with $D$, where $C = C_{\mathcal{A}}(g_2^{(t)} \circ g_1^{(t)}  \circ h^{(t)}, \epsilon, \delta)$.
\end{lem}
\begin{proof}
    The proof is given in Appendix \ref{app:lem_gnn_mlp}.
\end{proof}

With Lemma \ref{lem:gnn_mlp} at hand, by adopting Theorem \ref{thm:aa} and the generalization bound of neural networks \cite{arora2019fine,koltchinskii2002empirical}, we obtain a quantitative result for comparing the GNN architecture in \eqref{gnn:wcgcn} and the MLP architecture in the following theorem.

\begin{thm}\label{thm:gen_gnn_mlp} Under the setting of Lemma \ref{lem:gnn_mlp}, the following two conclusions hold:
    \begin{enumerate}
        \item For a given node number $|V|$, to achieve an error $\epsilon$ and a failure probability $\delta$, the MLP architecture requires $\mathcal{O}(|V|^2)$ times more samples than that for the GNN in \eqref{gnn:wcgcn};
        \item Consider a given training sample size $N$, a failure probability $\delta$, and a node number $|V|$. If GNN in \eqref{gnn:wcgcn} achieves an error of $\epsilon$ and $\delta > \Omega(\exp(-\epsilon))$, then the MLP architecture achieves an error of $\mathcal{O}(|V|\epsilon)$.
    \end{enumerate}
\end{thm}
\begin{proof}
    The proof is given in Appendix \ref{app:gen_gnn_mlp}.
\end{proof}
Theorem \ref{thm:gen_gnn_mlp} suggests that GNNs are superior to MLPs in learning DMP algorithms in terms of the sample complexity and generalization error, and the performance gain grows with the number of nodes in the graph. This verifies the advantages in the scalability and sample efficiency of GNNs over MLPs that have been empirically observed in \cite{shen2019graph,lee2019graph,eisen2019optimal,shen2020graph,guo2021learning}.

\section{Guidelines of Using Graph Neural Networks}\label{sec:arch}
The generalization analysis in the previous section not only theoretically demonstrates the advantages of GNNs over MLPs for solving wireless communication problems, but also can be used to design GNN architectures. In this section, we first introduce architecture design guidelines for performance improvement over the basic GNN architectures introduced in Section \ref{subsec:mpgnn}. We will then apply them to the problems in Section \ref{sec:model}. Finally, we will compare GNNs with classic methods and model-driven deep learning.
\subsection{General Guidelines for Theory-Guided Performance Enhancement}
To further improve the bound in \eqref{bound:C}, i.e., to reduce $C$, we should design GNN architectures to make $h^{(t)}$, $g_1^{(t)}$, and $g_2^{(t)}$ easier to learn. There are four commonly adopted approaches as listed below.

\begin{enumerate}
    \item \emph{To simplify the message passing scheme}: The complexity of the functions often increases with the input dimension, which depends on the number of incoming messages to a node. Thus, the complexity can be reduced if we can simplify the message-passing scheme. We will illustrate this approach in Section \ref{arch:cf}.
    \item \emph{To find a suitable pooling function}: The most widely adopted pooling functions are the max, sum, or attention aggregation. The max aggregation focuses on several neighbors that are most influential (Theorem 6 in \cite{shen2020graph}). It is suitable when the neighbors' influence is sparse, or the problem parameters are noisy. The sum aggregation gives summary statistics of the neighbors. It is useful when we aim to obtain a summary of the neighbors. The attention-based aggregation focuses on neighbors that have a large correlation with the current node. It can be used when the mutual correlations between the nodes are important \cite{li2021heterogeneous}.
    \item \emph{To improve the aggregation neural network $q_1(\cdot)$ and combination neural network $q_2(\cdot)$}: In most GNNs \cite{wang2019deep}, the architectures for $q_1$ and $q_2$ are one-layer MLPs. When a good classic DMP algorithm is available, one can design these neural networks via standard unrolling techniques \cite{he2019model,hu2020iterative}, as will be illustrated in Section \ref{arch:hb}.
    \item \emph{To apply input embedding}: When the dimension of the edge feature is much smaller than that of the hidden dimension, the neural network often ignores the edge feature. Under this circumstance, it is often useful to employ an embedding neural network that lifts the edge feature into a higher dimension.
\end{enumerate}

\subsection{Power Control in D2D Networks}\label{arch:d2d}
In this application, we apply the second and the fourth guidelines to enhance the performance with ECGNN as a basic architecture, which was introduced in Example \ref{exmp:ecgnn}. In the ECGNN, the dimensions of the node feature and edge feature are much smaller than the dimensions of the hidden layers. Hence, we use a one-layer MLP to embed the node feature and edge feature. As observed in \cite{sun2018learning}, the output power in the sum rate maximization problem is sparse. Thus, the influence from neighbors is sparse and better performance may be achieved with max aggregation. Specifically, we first embed the node feature and edge feature as
\begin{align}\label{gnn:pcgnn1}
    \bm{x}^{(0)}_k = \sigma(\bm{W}^{(0)}_1 h_{k,k}), \quad \bm{e}_{j,k} = \sigma(\bm{W}_2^{(0)}[h_{j,k}, h_{k,j}]),
\end{align}
where $\bm{W}_1^{(0)}$ and $\bm{W}_2^{(0)}$ are learnable weights. For the GNN design, we follow the message-passing scheme of ECGNN but replace the sum aggregation with a max aggregation. Then, the update is given by
\begin{equation}\label{gnn:pcgnn2}
    \begin{aligned}
    &\bm{d}_k^{(t)} = \max_{j\neq k}\sigma(\bm{W}_1^{(t)}\bm{x}_k^{(t)} + \bm{W}_2 \bm{e}_{j,k} +  \bm{W}_3^{(t)}\bm{x}_j^{(t)}), \\
    &p_k = \text{MLP}(\bm{d}_k^{(T)}),
\end{aligned}
\end{equation}
where $\bm{W}_1^{(t)}, \bm{W}_2^{(t)}$, and $\bm{W}^{(t)}_3$ are learnable weights at the $t$-th layer.

\subsection{Power Control in Cell-free Massive MIMO}\label{arch:cf}
In this subsection, we improve the neural architecture of HetGNN, as introduced in Example \ref{exmp:hetgnn}, by simplifying the message passing scheme. Notice that the interference in \eqref{eq:sinr_cf} can be expressed as:
\begin{align*}
&\sum_{k'\neq k}^K p_{k'}\underbrace{\left( \left(\sum_{m=1}^M v_{m,k}\frac{u_{m,k'}}{u_{m,k}} \right)^2 |\bm{\phi}_k^H\bm{\phi}_{k'}|^2 + v_{m,k}u_{m,k'}\right)}_{\text{interference channel from user $k'$}} \\ &+ \sum_{m=1}^M v_{m,k}u_{m,k}p_k + \frac{1}{\rho} \sum_{m=1}^M v_{m,k},
\end{align*}
which is similar to the interference term in the D2D networks in \eqref{prob:sum_rate}. This motivates us to introduce the message passing between users as in Section \ref{arch:d2d}. Moreover, as there are neither node features nor optimization variables on the AP nodes, we run the message passing from the APs to the users only in the first iteration to simplify the message passing. For the aggregation function, we choose the mean aggregation as the power values are not sparse in the fairness problem in \eqref{prob:cf}. Specifically, in the first iteration, the messages are passed from the APs to users and the update is given by
\begin{align}\label{gnn:cf_pcgnn1}
    &\bm{d}_k^{(1)} = \frac{1}{M} \sum_{m=1}^M \sigma(\bm{w}_1^{(1)} u_{m,k}), \quad 1 \leq m \leq M, M < k \leq M + K,
\end{align}
where $\bm{w}_1$ is a learnable weight vector, and $\sigma(\cdot)$ is an activation function. For $t \geq 2$, the message passing scheme follows HetGNN but the messages are passed only among the users but not the APs. Specifically, the update is
\begin{align}\label{gnn:cf_pcgnn2}
    \begin{aligned}
    &\bm{a}_k^{(t)} = \sum_{m = M+1}^{M+K} \bm{W}_1^{(t)} \bm{d}_m^{(t-1)}, \quad M < k \leq M + K, \\
    &\bm{d}_k^{(t)} = \sigma(\bm{W}_2^{(t)} \bm{d}_k^{(t-1)} + \bm{a}_k^{(t)}), \quad M < k \leq M + K, \\
    &p_k = \text{MLP}(\bm{d}^{(T)}_k), \quad M < k \leq M + K,
\end{aligned}
\end{align}
where $\bm{W}_1^{(t)}$, $\bm{W}_2^{(t)}$ are learnable weights at the $t$-th layer, $\sigma(\cdot)$ is the activation function, and $\text{MLP}(\cdot)$ is a learnable MLP mapping the hidden state to the power values.

\subsection{Partially Connected Hybrid Precoding}\label{arch:hb}
In this subsection, we design the GNN architecture for hybrid precoding with the partially connected structure, i.e., Problem \eqref{prob:hybrid_partial} via deep unrolling. As discussed in \cite{yu2016alternating}, when $\bm{F}_{\rm{BB}}$ is fixed, we can obtain a close form solution of $\bm{F}_{\rm{RF}}$ as follows:
\begin{align}\label{upd:frf}
    \bm{F}_{\rm{RF}} = P_{\mathcal{X}}( \bm{F}_{\rm{opt}} \bm{F}^H_{\rm{BB}}),
\end{align}
where $P_{\mathcal{X}}$ is a projection onto the constraints in \eqref{eq:partial}. When $\bm{F}_{\rm{RF}}$ is fixed, $\bm{F}_{\rm{BB}}$ should be obtained via solving a semi-definite programming. Note that \eqref{upd:frf} can be interpreted as the following message passing from the symbol nodes to antenna nodes:
\begin{align*}
    &\bm{m}_{j \rightarrow i} = (\bm{F}_{\rm{opt}})_{(i,j)}(\bm{F}_{\rm{BB}})_{(:,j)}^H = {\bf A}_{(j,i+N_t)}^* (\bm{X}_{\rm{BB}})_{(j,:)}, \\
    &1\leq i \leq N_t, 1 \leq j \leq N_s\\
    &(\bm{F}_{\rm{RF}})_{(i,:)} = P_{\mathcal{X}} \left(\sum_{j=1}^{N_s} \bm{m}_{j \rightarrow i}\right), \quad 1\leq i \leq N_t.
\end{align*}
Motivated by this message-passing scheme, we design the following GNN:
\begin{equation}\label{gnn:unroll_gnn}
\small
    \begin{aligned}
    &\bm{a}_k^{(t)} = \left\{
\begin{aligned}
&\sum_{j=N_t+1}^{N_t + N_s} {\bf A}_{(j,k+N_t)}^* (\bm{d}^{(t-1)}_{j}) = \sum_{j = 1}^{N_s} (\bm{F}_{\rm{opt}})_{(k,j)} (\bm{d}^{(t-1)}_{j}), \\
&1 \leq k \leq N_t \\
&\sum_{j=1}^{N_t} {\bf A}_{(k,j+N_s)}^* (\bm{d}^{(t-1)}_{j})\bm{W} = \sum_{j=1}^{N_t} (\bm{F}_{\rm{opt}})_{(j,k-N_t)} (\bm{d}^{(t-1)}_{j})\bm{W}, \\ &N_t < k \leq N_t + N_s,
\end{aligned}
\right. \\
    &\bm{d}_k^{(t)} = \left\{
\begin{aligned}
&P_{\mathcal{X}} (\bm{a}_k^{(t)}), && 1 \leq k \leq N_t \\
& \sqrt{\frac{N_{RF}N_s}{N_t}} P_{\mathbb{S}}(\bm{a}_k^{(t-1)}), && N_t < k \leq N_t + N_s,
\end{aligned}
\right. 
\end{aligned}
\end{equation}
where $\bm{W} \in \mathbb{C}^{N_{RF} \times N_{RF}}$ is learnable weights at the $t$-th layer and $P_{\mathbb{S}}(\bm{x}) = \frac{\bm{x}}{\|\bm{x}\|_2}$ is a projection to the sphere.

\subsection{Benefits and Limitations of Applying GNNs}
In this subsection, we discuss the pros and cons of applying GNNs in wireless communications by comparing them with other methods.

\paragraph{GNNs versus classic algorithms} Compared with classic algorithms, GNNs have shown their superior performance in large-scale resource management \cite{shen2020graph}, multi-agent semantic communication \cite{zhou2022Multi}, and MIMO detection \cite{kosasih2022graph}. Furthermore, once a GNN is trained, it becomes a DMP algorithm (Theorem \ref{thm:GNNDMP}) and can be deployed in a distributed manner \cite{wang2022learning}. Nevertheless, the performance guarantee of GNNs (Lemma \ref{lem:gnn_mlp}) requires the training distribution to be the same as the test distribution while their robustness to channel distribution shift remains elusive in theory. Thus, classic algorithms may be better choices when the testing environment differs substantially from the training setting (e.g., ITU $\rightarrow$ LTE in Table \ref{tab:gen_channel}).

\paragraph{GNNs versus other neural architectures} Compared with other neural networks, the unique advantages of GNNs are good scalability and generalization, which are desirable in future wireless networks that are ultra-dense and dynamically changing. As discussed before, GNNs are the only learning-based methods that can scale up to a large number of users \cite{shen2020graph,kosasih2022graph} and antennas \cite{he2022gblinks,kosasih2022graph}. In addition, GNNs enjoyed unreasonably good generalization in wireless communication problems. For resource management problems, previous studies demonstrated that GNNs are robust to the variations in system parameters (e.g., network density \cite{eisen2019optimal}, SNR level \cite{jiang2020learning}, the number of antennas \cite{kim2022bipartite,he2022gblinks}, and the number of subcarriers \cite{he2022gblinks}), which cannot be achieved by classic neural architectures. For statistical inference problems, GNNs are robust to the change of SNR level, number of antennas, and the number of users \cite{kosasih2022graph}.

\paragraph{GNNs versus model-driven deep learning} GNNs can be used in a data-driven manner \cite{shen2020graph,jiang2020learning,guo2021learning} or integrated with model-driven methods \cite{eisen2019optimal,chowdhury2021unfolding,he2022graph}. As the definition of model-driven deep learning is ambiguous, a thorough study goes beyond the scope of this paper. For simplicity, we only investigate the comparison between model-driven GNNs and data-driven ones in one common task, i.e., beamforming. Specifically, we consider the problem where the input $Y$ is a (random) function of the instantaneous CSI matrix $X$ and the output is the beamforming matrix. To use unrolled-based methods \cite{hu2020iterative,chowdhury2021unfolding}, we should first reconstruct $X$ and then feed it into the neural networks. When the mutual information $I(X,Y)$ is much smaller than entropy $H(X)$, it is difficult to reconstruct $X$ from $Y$ and data-driven methods are superior. This is the case when $Y$ is limited channel feedback \cite{sohrabi2021deep} or pilots \cite{jiang2020learning}. Otherwise, a proper combination of model-driven deep learning and GNNs can always improve the performance as it reduces $C$ in bound \eqref{bound:C}. Furthermore, data-driven methods require careful data normalization while model-driven ones do not.


\section{Simulations}\label{sec:exp}
In this section, we conduct simulations to verify our theoretical results in Section \ref{sec:aa} and the effectiveness of the design methodologies developed in Section \ref{sec:arch}.
\subsection{Power Control in D2D Networks}
\begin{figure*}[htbp]
    \centering
    \subfigure[Performance versus training samples, SNR = 10dB.]
    {
        \includegraphics[width=0.9\columnwidth]{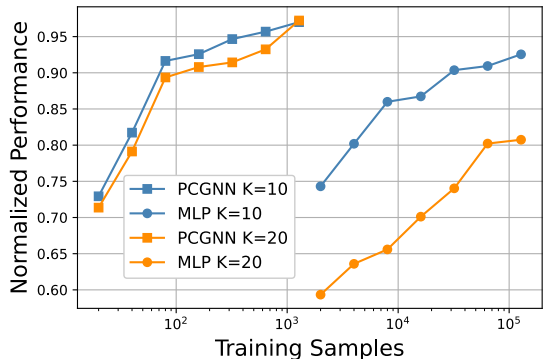}
        \label{fig:sample}
    }
    \subfigure[Performance versus transceiver pairs, SNR = 10dB.]
    {
        \includegraphics[width=0.9\columnwidth]{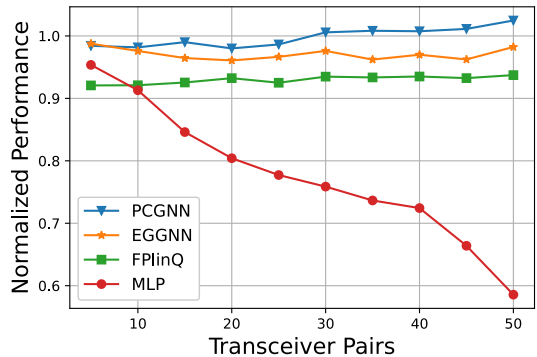}
        \label{fig:per_d2d}
    }
    \caption{Performance of different benchmarks for power control in D2D networks. The normalized performance refers to the sum rate of different methods normalized by the sum rate of \emph{Best FPlinQ}.}
\end{figure*}
We first investigate the sum rate maximization problem (i.e., Problem \eqref{prob:sum_rate}) in a single-antenna Gaussian interference channel, i.e., $h_{j,k} \sim \mathcal{CN}(0,1)$ and $w_k = 1, \forall k$, the same setting as \cite{sun2018learning,liang2018towards}. We consider the following benchmarks for comparison. 
\begin{enumerate}
    \item \textbf{FPlinQ} \cite{shen2017fplinq}: FPLinQ is a state-of-the-art optimization-based algorithm for power control in D2D networks. 
    \item \textbf{MLP} \cite{sun2018learning}: This is a classic learning-based method for power control in D2D networks and we follow the setting of \cite{sun2018learning}. 
    \item \textbf{ECGNN} \cite{shen2019graph}: ECGNN is a GNN-based method for power control as is shown in \eqref{gnn:ecgnn} and we set the number of layers as $2$. 
    \item \textbf{PCGNN}: PCGNN is the proposed GNN in \eqref{gnn:pcgnn1} and \eqref{gnn:pcgnn2} and the number of layer is $2$. 
    \item \textbf{UWMMSE} \cite{chowdhury2021unfolding}: This method unfolds the WMMSE algorithm and the learnable parameters are parameterized with GNNs. The GNN's width and depth, and the number of training samples are chosen as suggested in \cite{chowdhury2021unfolding}.
    \item \textbf{Best FPlinQ}: For each test sample, we run FPlinQ with $100$ different initialization points and use the best one as the final performance. The sum rate achieved by this method is often regarded as the highest achievable one \cite{hu2020iterative,shen2020graph,chowdhury2021unfolding}. 
\end{enumerate}

We first compare the sample complexity of MLPs and GNNs when the number of users is $10$ or $20$ and SNR = 10dB, as shown in Fig. \ref{fig:sample}. For $K=10$, PCGNN trained with $20$ and $40$ training samples achieves similar performance with the MLP trained with $2,000$ and $4,000$ training samples, respectively. Remarkably, for $K=20$, PCGNN trained with $20$ and $40$ training samples achieves similar performance with MLP trained with $16,000$ and $64,000$ training samples, respectively. This verifies that GNNs need lower sample complexity than MLPs approximately in the order of $\mathcal{O}(K^2)$. We also test the performance of different methods versus the number of transceiver pairs. The results are shown in Fig. \ref{fig:per_d2d}, where SNR = $10$dB and we use $100,000$ training samples for learning-based methods. It is observed that the performance of PCGNN and ECGNN is stable while that of MLPs decreases linearly with the network size. This verifies the statement in Theorem \ref{thm:gen_gnn_mlp} that GNNs achieve a better performance than MLPs approximately in the order of $\mathcal{O}(K)$. In addition, the proposed PCGNN outperforms ECGNN consistently, which demonstrates the effectiveness of our design guidelines. 

Surprisingly, PCGNN even outperforms Best FPlinQ when the number of users is large. This is because the number of local minimizers grows \emph{exponentially} with the user number and optimization-based methods tend to be trapped in bad local minimizers. Intuitively, the neural networks first implicitly model the optimization landscape during training and generate solutions according to the landscape in the test. From this perspective, neural networks share a similar spirit as Gibbs algorithms \cite{qian2012distributed}, which are DMP algorithms and robust to find a better local minimizer. Nevertheless, neural networks are much more powerful in large-scale non-convex optimization as they can efficiently model distributions.

In the following, we test the robustness of PCGNN. We consider $10$ single-antenna transceiver pairs within a $250 \times 250\text{m}^2$ area and the channel model is ITU-1411 \cite{itu1411}. The transmitters are randomly located in this area while each receiver is uniformly distributed within $[d_{\text{min}},d_{\text{max}}]$ from the corresponding transmitter.
\paragraph*{Robustness to SNR level} We set $d_{\text{min}} = 10\text{m}$ and $d_{\text{max}} = 50\text{m}$. In the training dataset, the transmit power is uniformly distributed in $10-30\text{dBm}$ and we use transmit power as the node feature \cite{liang2018towards}. We test the performance of GNNs across different transmit power values, which is shown in Table \ref{tab:gen_snr}. The performance is normalized by Best FPlinQ and ``PCGNN (Full training)'' refers to the PCGNN whose training transmit power is the same as the test transmit power. The simulations demonstrate that PCGNN's performance is stable in terms of the SNR level.

\begin{table}[htb]
	
	\selectfont  
	\centering
		\newcommand{\tabincell}[2]{\begin{tabular}{@{}#1@{}}#2\end{tabular}}
	\caption{Robustness to SNR level. $K=10$.} 
	\resizebox{0.42\textwidth}{!}{
		\begin{tabular}{|c|c|c|c|c|c|c|}  
			\hline  
			\tabincell{c}{Transmit \\Power (dBm)}&         PCGNN       &     \tabincell{c}{PCGNN \\(Full Training)}     &      \tabincell{c}{FPlinQ}    \\ \hline
			  0        &      $96.12\%$      &       $97.10\%$       &      $93.43\%$           \\ \hline
			10  &      $96.80\%$      &        $97.20\%$       &      $93.75\%$          \\\hline
			  20  &      $97.00\%$      &       $97.39\%$      &      $93.62\%$        \\\hline
                30  &      $96.65\%$      &       $97.12\%$      &      $93.31\%$       \\\hline
                40  &      $96.90\%$      &       $97.13\%$      &      $93.18\%$        \\\hline

	\end{tabular}}
	\label{tab:gen_snr}
\end{table}

\paragraph*{Robustness to channel distribution shifts} As empirically verified in \cite{shen2020graph,eisen2019optimal,jiang2020learning}, the performance of GNN-based methods is robust when the channels change. We further test four extreme situations here: a) \textbf{User distribution shift}: We set $d_{\text{max}} = d_{\text{min}} = 30\text{m}$ in the training while $d_{\text{max}} = 50\text{m}$ and $d_{\text{min}} = 10\text{m}$ in the test; b) \textbf{Antenna height distribution shift}: In the training, the height of both transmit antenna and receiver antenna is $1.5\text{m}$. In the test, the height of transmit antenna is uniformly distributed in $30-50\text{m}$ and the receiver antenna is uniformly distributed in $1-3\text{m}$; c) \textbf{LoS $\rightarrow$ NLoS}: There are only line-of-sight paths in the training while both line-of-sight and non-line-of-sight paths exist in the test; d) \textbf{ITU $\rightarrow$ LTE}: The channel models in the training is ITU-1411 \cite{itu1411} while we follow channel model in the standard LTE cellular network \cite{insoo13holistic} in the test. The simulations are shown in Table \ref{tab:gen_channel}. The performance is normalized by Best FPlinQ and ``PCGNN (Full training)'' refers to the PCGNN whose training channel distribution is the same as the test ones. Interestingly, we see that the performance loss is subtle in the first three tasks even if the distribution shift is significant, which demonstrates the robustness of GNNs when the channel changes. The performance drop in the fourth task is because the training data and test data have different supports, which shows that retraining is needed when there are significant variations in channel models.

\begin{table}[htb]
	
	\selectfont  
	\centering
	\newcommand{\tabincell}[2]{\begin{tabular}{@{}#1@{}}#2\end{tabular}}
	\caption{Robustness to the change of channels. $K=10$.} 
	
	\resizebox{0.48\textwidth}{!}{
		\begin{tabular}{|c|c|c|c|c|c|}  
			\hline  
			Setting &     PCGNN      &    \tabincell{c}{PCGNN \\(Full Training)}      &       \tabincell{c}{FPlinQ}              \\ \hline
			User distribution shift &      $97.62\%$     &      $97.78\%$      &             $93.51\%$              \\ \hline
			\tabincell{c}{Antenna height\\ distribution shift} &       $97.63\%$     &      $97.63\%$      &       $93.66\%$         \\ \hline
			LoS $\rightarrow$ NLoS     &      $94.43\%$      &      $97.53\%$      &      $93.51\%$           \\\hline
			ITU $\rightarrow$ LTE  &      $88.71\%$      &       $96.57\%$      &       $94.32\%$           \\\hline
			
	\end{tabular}}
	\label{tab:gen_channel}
\end{table}

\begin{table*}[htb]
	
	\selectfont  
	\centering
	
	\caption{Average minimum rate achieved by different benchmarks divided by the average minimum rate achieved by the \emph{optimal algorithm}.} 
	\newcommand{\tabincell}[2]{\begin{tabular}{@{}#1@{}}#2\end{tabular}}
	\resizebox{0.75\textwidth}{!}{
		\begin{tabular}{|c|c|c|c|c|c|c|c|c|}
			\hline
			$M$& $K$  & Maximum Power & MLP  & \tabincell{c}{HetGNN \\$K_{\text{train}} = 6$} & \tabincell{c}{CF-PCGNN \\$K_{\text{train}} = 6$} &    HetGNN & CF-PCGNN  \cr \hline
			\multirow{3}{*}{30} & 6  &$46.3\%$  & $84.4\%$ & $94.9\%$ & $95.1\%$ & $94.9\%$ & $95.1\%$  \cr \cline{2-8} 
			& 8 & $35.8\%$ & $74.8\%$  &$90.2\%$  & $92.1\%$ & $91.9\%$ & $93.1\%$ \cr \cline{2-8} 
			& 10 & $29.4\%$ & $70.5\%$  & $85.7\%$ & $90.3\%$ & $90.8\%$  & $91.4\%$  \cr \hline
			\multirow{3}{*}{50} & 6 & $63.7\%$ & $90.4\%$ & $96.2\%$ & $96.8\%$ & $96.2\%$ & $96.8\%$  \cr \cline{2-8} 
			& 8 & $55.5\%$ & $84.2\%$ & $94.4\%$ & $95.2\%$ & $95.6\%$ & $95.6\%$\cr \cline{2-8} 
			& 10 & $46.8\%$  & $78.7\%$ & $93.1\%$ & $94.5\%$ & $94.7\%$& $95.1\%$\cr \hline
	\end{tabular}}
	\label{tab:cf}
\end{table*}

\begin{figure*}[htbp]
    \centering
    \subfigure[The energy efficiency achieved by different benchmarks.]
    {
        \includegraphics[width=0.8\columnwidth]{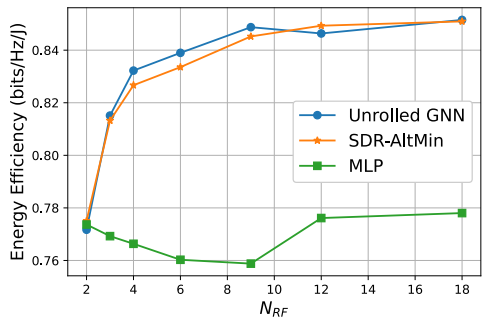}
        \label{fig:hb_ee}
    }
    \subfigure[The computation time of different benchmarks.]
    {
        \includegraphics[width=0.8\columnwidth]{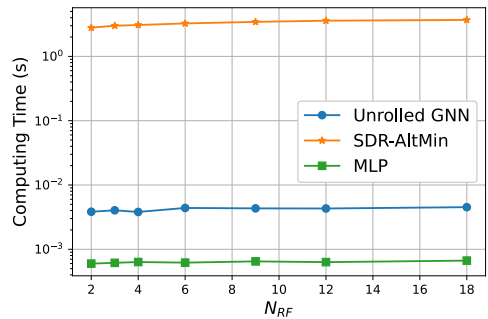}
        \label{fig:hb_ct}
    }
    \caption{Performance and computation time in hybrid precoding. }
    \label{fig:cf_per}
\end{figure*}

\subsection{Power Control in Cell-free Massive MIMO}
In this subsection, we consider the power control problem formulated in \eqref{prob:cf} in a cell-free network in an area of $0.5 \times 0.5 \text{km}^2$ with different numbers of APs $M$ and users $K$. Other system configurations are identical to \cite{ngo2017cell,rajapaksha2021deep}. We consider the following benchmarks for comparison. 
\begin{enumerate}
    \item \textbf{Maximum Power}: All users transmit with full power.
    \item \textbf{MLP} \cite{rajapaksha2021deep}: It is developed in \cite{rajapaksha2021deep} for the uplink power control in cell-free massive MIMO. 
    \item \textbf{HetGNN}: HetGNN is a GNN architecture designed for the bipartite graph. We modify it for the power control problem in cell-free massive MIMO according to the guidelines in \cite{guo2021learning}. The updates of HetGNN are shown in \eqref{gnn:hetgnn} and the number of layers is set as $3$.
    \item \textbf{CF-PCGNN} \cite{shen2019graph}: CF-PCGNN is the proposed GNN in \eqref{gnn:cf_pcgnn1} and \eqref{gnn:cf_pcgnn2} and the number of layers is set as $2$.
    \item \textbf{Optimal Algorithm}: This algorithm knows all the needed information, i.e., $v_{m,k}$, $u_{m,k}$, $\rho$, and $\bm{\phi}_k$ in \eqref{eq:sinr_cf}. Then an optimal bisection algorithm is adopted to solve Problem \eqref{prob:cf}. It will be used as the performance upper bound. 
\end{enumerate}
For the learning-based methods, the input is the large-scale fading coefficient $v_{m,k}$. We consider $M = 30,50$ and $K = 6,8,10$. For each setting, we generate $10,000$ samples for training and $2,000$ samples for testing. The performance of different benchmarks is shown in Table \ref{tab:cf}. In the table, we also test the ability of GNNs to generalize to different settings. Specifically, we train HetGNN and CF-PCGNN on a network with $6$ users and test them on networks with $6-10$ users, which we refer to as ``HetGNN $K_{\text{train}} = 6$'' and ``CF-PCGNN $K_{\text{train}} = 6$''. As shown in the table, the performance of GNNs is superior to that of MLPs, and the performance gain becomes larger with a larger number of users. In addition, the proposed GNN outperforms HetGNN, which demonstrates the effectiveness of our neural architecture design guidelines. In this task, we notice that the performance of GNNs degrades when the number of users grows. This is because the objective function is not differentiable and training becomes unstable when there are more variables. This can be alleviated by adopting a smooth envelope of the loss function, but it is beyond the scope of this paper.

\subsection{Partially Connected Hybrid Precoding}
This subsection compares different methods for the hybrid precoding problem with a partially connected structure. The data streams are sent from a transmitter with $N_t = 144$ antennas to a receiver with $N_r = 36$ antennas, and SNR = 10dB. The other system parameters are identical with \cite{yu2016alternating}. We consider the following benchmarks: 
\begin{enumerate}
    \item \textbf{MLP}: An MLP is adopted to learn the map from $\bm{F}_{\rm{opt}}$ to the optimal $\bm{F}_{\rm{RF}}$ and then $\bm{F}_{\text{BB}}$ is obtained by \eqref{upd:frf}.
    \item \textbf{Unrolled GNN}: It is the proposed GNN in \eqref{gnn:unroll_gnn} and we unfold $10$ iterations. 
    \item \textbf{SDR-AltMin}: SDR-AltMin is an optimization-based method for partially connected hybrid precoding developed in \cite{yu2016alternating}.
\end{enumerate}
We use energy efficiency as the metric to evaluate the performance of different benchmarks, which is defined as the ratio between the spectral efficiency and total power consumption:
\begin{align*}
    \eta = \frac{R}{P_{\rm{common}} + N_{\rm{RF}}P_{\rm{RF}} + N_t(N_{\rm{PA}} + N_{\rm{PS}})}
\end{align*}
where $R$ is the spectral efficiency, $P_{\text{common}} = 10\rm{W}$, $P_{\rm{RF}} = 100\rm{mW}$, $P_{\rm{PS}} = 10\rm{mW}$, and $P_{\rm{PA}} = 100\rm{mW}$ \cite{yu2016alternating}. We generate $10,000$ samples for training and $2,000$ samples for testing.

The performance results are shown in Fig. \ref{fig:hb_ee}. A unique property of the partially connected structure is that the energy efficiency increases with the number of RF chains \cite{yu2016alternating}. This trend can be revealed via SDR-AltMin or GNN-based approaches but not MLP-based approaches. In addition, the performance of unrolled GNN significantly outperforms that of MLPs. We also compare the computation time between the learning-based methods and optimization-based methods. For a fair comparison, we test all the methods on Intel(R) Xeon(R) CPU @ 2.20GHz and set the test batch size as $1$ for the learning-based methods. The results are shown in Fig. \ref{fig:hb_ct} and we see that the learning-based methods achieve hundreds of times speedups compared with the optimization-based SDR-AltMin.

\section{Conclusions}\label{sec:con}
In this paper, we developed a powerful framework for applying GNNs in wireless networks, consisting of graph modeling, neural architecture design, and theoretical analysis. In contrast to existing learning-based methods, we focused on the theoretically principled approaches to meet the key performance requirements, including scalability, good generalization, and high computational efficiency. Moreover, this paper provided the first quantitative result regarding the performance gain of GNNs over classic neural architectures in wireless communication systems, which casts light on the performance analysis and principled neural architecture design of deep learning-based approaches. As for future directions, it will be interesting to apply our framework to more design problems in wireless communications. It will also be interesting to combine GNNs with model-driven deep learning approaches and extend our theoretical results to this setting. 

\appendices
\section{Proof of Theorem \ref{fact:dmp}}\label{app:dmp}
The notations of this section follow \cite{shen2020graph} and we will not introduce them due to space limitation. We prove Theorem \ref{fact:dmp} by showing that the map from the problem parameter to the optimal solution, referred as the \emph{optimal map} $F:\bm{Z}, {\bf A} \mapsto \bm{X}^*$, can be expressed as a DMP algorithm. First, by Proposition 4 of \cite{shen2020graph}, we know that the optimal map is permutation equivariant, i.e., $F(\pi \star \bm{Z}, \pi \star {\bf A}) = \pi \star F(\bm{Z}, {\bf A})$. Let $F = [f_1, \cdots, f_n]$, where $f_i: \bm{Z}, {\bf A} \mapsto (\bm{X}^*)_{(i,:)}$. From a node's view, the optimal map from the problem parameter to its optimal solution is permutation invariant, i.e., $f_{\pi(i)}(\pi \star \bm{Z}, \pi \star {\bf A}) = f_{i} (\pi \star \bm{Z}, \pi \star {\bf A})$. From Theorem 7 of \cite{zaheer2017deep}, any permutation invariant function $f(x_1, \cdots, x_n)$ can be decomposed as $f(x_1,\cdots, x_n) = \psi\left(\sum_{i=1}^{n} \phi(x_i)\right)$, where $\psi:\mathbb{R}^{n+1} \rightarrow \mathbb{R}$, $\phi: \mathbb{R} \rightarrow \mathbb{R}^{n+1}$. This decomposition can be implemented with a DMP algorithm where $\bm{m}_{j \rightarrow i} = \phi(\bm{x}_j^{(t-1)}),  \bm{y}_i^{(t)} = \sum_{j} \bm{m}_{j\rightarrow i}^{(t)}, \bm{x}_i^{(t)} = \bm{y}_i^{(t)}$. Thus, for any graph optimization problem in \eqref{eq:cg_opt}, there exists a DMP algorithm that can solve it.

\section{Proof of Lemma \ref{lem:gnn_mlp}}\label{app:lem_gnn_mlp}
We first check the second condition of AA for the GNN in \eqref{gnn:wcgcn}. The basic module of this GNN is shown in \eqref{eq:align} and the training algorithm $\mathcal{A}$ is gradient descent. Thus, a sample complexity bound on MLP trained via gradient descent is needed. Following Theorem 5.1 of \cite{arora2019fine} or Lemma \ref{lem:mlp_gen}, we can ensure the existence of $C_{\mathcal{A}}(h^{(t)},\epsilon, \delta), C_{\mathcal{A}}(g_2^{(t)} \circ g_1^{(t)},\epsilon, \delta) < \infty$, where $\mathcal{A}$ is the gradient descent algorithm. The second condition is satisfied if we set $N > 2T \cdot \frac{C}{|V|}$,and $C = \max_{t = 1, \cdots, T} (C_{\mathcal{A}}(h^{(t)},\epsilon, \delta), C_{\mathcal{A}}(g_1^{(t)} \circ g_2^{(t)},\epsilon, \delta))$. The $|V|$ term in the denominator is because each training sample is reused $|V|$ times as each node shares an identical MLP. 

To use AA for MLPs, we define the module functions $f_1, \cdots, f_{T|V|}$ and neural network modules $\mathcal{M}_1, \cdots, \mathcal{M}_{T|V|}$ as $f_i = g_2^{(\lfloor i/|V| \rfloor)} \circ g_1^{(\lfloor i/|V| \rfloor)}  \circ h^{(\lfloor i/|V| \rfloor)}, \quad \mathcal{M}_i = \text{MLP}_i$, where the definition of $\mathcal{M}_i$ is because part of an MLP is also an MLP. With these definitions, the first condition of AA holds, and the second condition is satisfied if we set $N > |V|T \cdot C , C = \max_{i=1,\cdots,T} C_{\mathcal{A}}(g_2^{(t)} \circ g_1^{(t)}  \circ h^{(t)}, \epsilon, \delta)$.

\section{Proof for Theorem \ref{thm:gen_gnn_mlp}}\label{app:gen_gnn_mlp}
We first present a useful lemma. Denoting the input as $\bm{x} \in \mathbb{R}^d$, we define deep function class recursively as follows:
\begin{equation}\label{eq:dnn}
    \begin{aligned}
    &\mathcal{F}^{(1)} = \{\bm{W}^T\bm{x}, \bm{W} \in \mathcal{W}\}, \\
    &\mathcal{F}^{(k)} = \mathcal{F}_{A_k, L_1}(h \circ \mathcal{F}^{(k-1)}) \\
    &= \left\{ \sum_{j=1}^m w_j h(f_j(x)): \|\bm{w}\|_1 \leq A_k, f_j \in \mathcal{F}^{(k-1)}\right\}.
\end{aligned}
\end{equation}

We have the following generalization bound for MLPs:
\begin{lem}\label{lem:mlp_gen} (Theorem 13 in \cite{koltchinskii2002empirical}) Consider $K$-layer MLPs defined by \eqref{eq:dnn}. Assume the output of neural networks is bounded by $M$ and $h$ is Lipschitz. Then there exists a constant $C$ for all distribution $\mathcal{D}$ such that the generalization error $\epsilon_n(\delta)$ is bounded by
\begin{align*}
    CM \left[ A \sqrt{\frac{d+1}{n}} + \sqrt{\frac{\log(1/\delta)}{n}} \right]= \mathcal{O}(1/\sqrt{n})
\end{align*}
where $n$ is the number of training samples and $A = 2^{K-2} \Pi_{l=2}^k A_l$.
\end{lem}

The first conclusion of Theorem \ref{thm:gen_gnn_mlp} directly follows Lemma \ref{lem:gnn_mlp} and Theorem \ref{thm:aa}. To obtain the second conclusion of Theorem \ref{thm:gen_gnn_mlp}, a generalization bound of the neural network is further required to transform the sample complexity into a generalization error. As shown in Lemma \ref{lem:mlp_gen}, the generalization error is inversely proportional to the square root of the sample complexity. As MLP's sample complexity is $|V|^2$ times larger than that of GNN in \eqref{gnn:wcgcn}, the generalization error of MLP is $|V|$ times larger than the GNN.

\bibliographystyle{ieeetr}
\bibliography{ref}

\begin{thebibliography}{10}

\bibitem{shen2021neural}
Y.~Shen, J.~Zhang, and K.~B. Letaief, ``How neural architectures affect deep
  learning for communication networks?,'' in {\em Proc. IEEE Int. Conf.
  Commun.}, pp.~1--6, Seoul, South Korea, May 2022.

\bibitem{sun2018learning}
H.~Sun, X.~Chen, Q.~Shi, M.~Hong, X.~Fu, and N.~D. Sidiropoulos, ``Learning to
  optimize: Training deep neural networks for interference management,'' {\em
  IEEE Trans. Signal Process.}, vol.~66, pp.~5438 -- 5453, Oct. 2018.

\bibitem{shen2020graph}
Y.~Shen, Y.~Shi, J.~Zhang, and K.~B. Letaief, ``Graph neural networks for
  scalable radio resource management: Architecture design and theoretical
  analysis,'' {\em IEEE J. Sel. Areas Commun.}, vol.~39, pp.~101--115, Jan.
  2021.

\bibitem{he2020model}
H.~He, C.-K. Wen, S.~Jin, and G.~Y. Li, ``Model-driven deep learning for {MIMO}
  detection,'' {\em IEEE Trans. Signal Process.}, vol.~68, pp.~1702--1715, Feb.
  2020.

\bibitem{shao2021learning}
J.~Shao, Y.~Mao, and J.~Zhang, ``Learning task-oriented communication for edge
  inference: An information bottleneck approach,'' {\em IEEE J. Sel. Areas
  Commun.}, vol.~40, pp.~197--211, Jan. 2022.

\bibitem{liang2018towards}
F.~Liang, C.~Shen, W.~Yu, and F.~Wu, ``Towards optimal power control via
  ensembling deep neural networks,'' {\em IEEE Trans. Commun.}, vol.~68,
  pp.~1760--1776, Mar. 2020.

\bibitem{he2019model}
H.~He, S.~Jin, C.-K. Wen, F.~Gao, G.~Y. Li, and Z.~Xu, ``Model-driven deep
  learning for physical layer communications,'' {\em IEEE Wireless Commun.},
  vol.~26, pp.~77--83, Oct. 2019.

\bibitem{shen2018lora}
Y.~Shen, Y.~Shi, J.~Zhang, and K.~B. Letaief, ``{LORM}: {Learning} to optimize
  for resource management in wireless networks with few training samples,''
  {\em IEEE Trans. Wireless Commun.}, vol.~19, pp.~665--679, Jan. 2020.

\bibitem{ma2021neural}
Y.~Ma, Y.~Shen, X.~Yu, J.~Zhang, S.~Song, and K.~B. Letaief, ``Neural
  calibration for scalable beamforming in {FDD} massive {MIMO} with implicit
  channel estimation,'' in {\em Proc. IEEE Global Commun. Conf.}, pp.~1--6,
  Madrid, Spain, Dec. 2021.

\bibitem{letaief2019roadmap}
K.~B. Letaief, W.~Chen, Y.~Shi, J.~Zhang, and Y.-J.~A. Zhang, ``The roadmap to
  {6G--AI} empowered wireless networks,'' {\em IEEE Commun. Mag.}, vol.~57,
  pp.~84--90, Aug. 2019.

\bibitem{letaief2022Edge}
K.~B. Letaief, Y.~Shi, J.~Lu, and J.~Lu, ``Edge artificial intelligence for
  {6G}: Vision, enabling technologies, and applications,'' {\em IEEE J. Sel.
  Areas Commun.}, vol.~40, pp.~5--36, Jan. 2022.

\bibitem{lee2019graph}
M.~Lee, G.~Yu, and G.~Y. Li, ``Graph embedding based wireless link scheduling
  with few training samples,'' {\em IEEE Trans. Wireless Commun.}, vol.~20,
  pp.~2282--2294, Apr. 2021.

\bibitem{eisen2019optimal}
M.~Eisen and A.~Ribeiro, ``Optimal wireless resource allocation with random
  edge graph neural networks,'' {\em IEEE Trans. Signal Process.}, vol.~68,
  pp.~2977--2991, Apr. 2020.

\bibitem{jiang2020learning}
T.~Jiang, H.~V. Cheng, and W.~Yu, ``Learning to beamform for intelligent
  reflecting surface with implicit channel estimate,'' {\em IEEE J. Sel. Areas
  Commun.}, vol.~39, pp.~1931--1945, Jul. 2021.

\bibitem{kosasih2022graph}
A.~Kosasih, V.~Onasis, V.~Miloslavskaya, W.~Hardjawana, V.~Andrean, and
  B.~Vucetic, ``Graph neural network aided {MU-MIMO} detectors,'' {\em IEEE J.
  Sel. Areas Commun.}, vol.~40, pp.~2540--2555, Sep. 2022.

\bibitem{chowdhury2021unfolding}
A.~Chowdhury, G.~Verma, C.~Rao, A.~Swami, and S.~Segarra, ``Unfolding {WMMSE}
  using graph neural networks for efficient power allocation,'' {\em IEEE
  Trans. Wireless Commun.}, vol.~20, pp.~6004--6017, Sep. 2021.

\bibitem{lee2021learning}
H.~Lee, S.~H. Lee, and T.~Q. Quek, ``Learning autonomy in management of
  wireless random networks,'' {\em IEEE Trans. Wireless Commun.}, vol.~20,
  pp.~8039--8053, Dec. 2021.

\bibitem{he2022graph}
H.~He, A.~Kosasihy, X.~Yu, J.~Zhang, S.~Song, W.~Hardjawanay, and K.~B.
  Letaief, ``Graph neural network enhanced approximate message passing for
  {MIMO} detection,'' {\em arXiv:2205.10620.}, 2022. [Online]. Available:
  https://arxiv.org/abs/2205.10620.

\bibitem{he2021Overview}
S.~He, S.~Xiong, Y.~Ou, J.~Zhang, J.~Wang, Y.~Huang, and Y.~Zhang, ``An
  overview on the application of graph neural networks in wireless networks,''
  {\em IEEE Open J. the Commun. Society}, vol.~2, pp.~2547--2565, Nov. 2021.

\bibitem{zhou2022Multi}
Y.~Zhou, J.~Xiao, Y.~Zhou, and G.~Loianno, ``Multi-robot collaborative
  perception with graph neural networks,'' {\em IEEE Robotics Automation
  Lett.}, vol.~7, pp.~2289--2296, Apr. 2022.

\bibitem{guo2021learning}
J.~Guo and C.~Yang, ``Learning power allocation for multi-cell-multi-user
  systems with heterogeneous graph neural network,'' {\em IEEE Trans. Wireless
  Commun.}, vol.~21, pp.~884--897, Feb. 2022.

\bibitem{valiant1984theory}
L.~G. Valiant, ``A theory of the learnable,'' {\em Commun. ACM}, vol.~27,
  pp.~1134--1142, Nov. 1984.

\bibitem{shen2019graph}
Y.~Shen, Y.~Shi, J.~Zhang, and K.~B. Letaief, ``A graph neural network approach
  for scalable wireless power control,'' in {\em Proc. IEEE Global Commun.
  Conf. Workshops}, pp.~1--6, Waikoloa, HI, USA, Dec. 2019.

\bibitem{angluin1980local}
D.~Angluin, ``Local and global properties in networks of processors,'' in {\em
  Proc. ACM Symp. Theory Comput.}, pp.~82--93, Los Angeles, CA, USA, Apr. 1980.

\bibitem{Shi2011An}
Q.~Shi, M.~Razaviyayn, Z.~Luo, and C.~He, ``An iteratively weighted {MMSE}
  approach to distributed sum-utility maximization for a {MIMO} interfering
  broadcast channel,'' {\em IEEE Trans. Signal Process.}, vol.~59,
  pp.~4331--4340, Sept. 2011.

\bibitem{yu2016alternating}
X.~Yu, J.-C. Shen, J.~Zhang, and K.~B. Letaief, ``Alternating minimization
  algorithms for hybrid precoding in millimeter wave {MIMO} systems,'' {\em
  IEEE J. Sel. Topics Signal Process.}, vol.~10, pp.~485--500, Apr. 2016.

\bibitem{yedidia2000bethe}
J.~S. Yedidia, W.~T. Freeman, and Y.~Weiss, ``Generalized belief propagation,''
  {\em Proc. Adv. Neural Inform. Process. Syst.}, vol.~13, pp.~689--696,
  Denver, CO, USA, Dec. 2000.

\bibitem{xu2019what}
K.~Xu, J.~Li, M.~Zhang, S.~Du, K.~Kawarabayashi, and S.~Jegelka, ``What can
  neural networks reason about?,'' in {\em Proc. Int. Conf. Learning
  Representations}, Apr. 2020. [Online]. Available:
  https://openreview.net/forum?id=rJxbJeHFPS.

\bibitem{ngo2017cell}
H.~Q. Ngo, A.~Ashikhmin, H.~Yang, E.~G. Larsson, and T.~L. Marzetta,
  ``Cell-free massive {MIMO} versus small cells,'' {\em IEEE Trans. Wireless
  Commun.}, vol.~16, pp.~1834--1850, Mar. 2017.

\bibitem{rajapaksha2021deep}
N.~Rajapaksha, K.~Manosha, N.~Rajatheva, and M.~Latva-aho, ``Deep
  learning-based power control for cell-free massive {MIMO} networks,'' in {\em
  Proc. Int. Conf. Commun.}, pp.~1--7, Montreal, Canada, May 2021.

\bibitem{wang2019deep}
M.~Wang, D.~Zheng, Z.~Ye, Q.~Gan, M.~Li, X.~Song, J.~Zhou, C.~Ma, L.~Yu,
  Y.~Gai, {\em et~al.}, ``Deep graph library: A graph-centric,
  highly-performant package for graph neural networks,'' {\em
  arXiv:1909.01315}, 2019. [Online]. Available:
  https://arxiv.org/abs/1909.01315.

\bibitem{loukas2019graph}
A.~Loukas, ``What graph neural networks cannot learn: depth vs width,'' in {\em
  Proc. Int. Conf. Learning Representations}, Apr. 2020. [Online]. Available:
  https://openreview.net/forum?id=B1l2bp4YwS.

\bibitem{hornik1989multilayer}
K.~Hornik, M.~Stinchcombe, and H.~White, ``Multilayer feedforward networks are
  universal approximators,'' {\em Neural Netw.}, vol.~2, no.~5, pp.~359--366,
  1989.

\bibitem{arora2019fine}
S.~Arora, S.~Du, W.~Hu, Z.~Li, and R.~Wang, ``Fine-grained analysis of
  optimization and generalization for overparameterized two-layer neural
  networks,'' in {\em Proc. Int. Conf. Mach. Learning}, pp.~322--332, Long
  Beach, CA, Jul. 2019.

\bibitem{koltchinskii2002empirical}
V.~Koltchinskii and D.~Panchenko, ``Empirical margin distributions and bounding
  the generalization error of combined classifiers,'' {\em Ann. Stat.},
  vol.~30, pp.~1--50, Feb. 2002.

\bibitem{li2021heterogeneous}
Y.~Li, Z.~Chen, Y.~Wang, C.~Yang, and Y.-C. Wu, ``Heterogeneous transformer: A
  scale adaptable neural network architecture for device activity detection,''
  {\em arXiv:2112.10086}, 2021. [Online]. Available:
  https://arxiv.org/abs/2112.10086.

\bibitem{hu2020iterative}
Q.~Hu, Y.~Cai, Q.~Shi, K.~Xu, G.~Yu, and Z.~Ding, ``Iterative algorithm induced
  deep-unfolding neural networks: Precoding design for multiuser {MIMO}
  systems,'' {\em IEEE Trans. Wireless Commun.}, vol.~20, pp.~1394--1410, Feb.
  2021.

\bibitem{wang2022learning}
Z.~Wang, M.~Eisen, and A.~Ribeiro, ``Learning decentralized wireless resource
  allocations with graph neural networks,'' {\em IEEE Trans. Signal Process.},
  vol.~70, pp.~1850--1863, Mar. 2022.

\bibitem{he2022gblinks}
S.~He, S.~Xiong, W.~Zhang, Y.~Yang, J.~Ren, and Y.~Huang, ``{GBLinks}:
  {GNN-based} beam selection and link activation for ultra-dense {D2D} {mmWave}
  networks,'' {\em IEEE Trans. on Commun.}, vol.~70, pp.~3451--3466, May 2022.

\bibitem{kim2022bipartite}
J.~Kim, H.~Lee, S.-E. Hong, and S.-H. Park, ``A bipartite graph neural network
  approach for scalable beamforming optimization,'' {\em IEEE Trans. Wireless
  Commun.}, early access, 2022.

\bibitem{sohrabi2021deep}
F.~Sohrabi, K.~M. Attiah, and W.~Yu, ``Deep learning for distributed channel
  feedback and multiuser precoding in {FDD} massive {MIMO},'' {\em IEEE Trans.
  Wireless Commun.}, vol.~20, pp.~4044--4057, Jul. 2021.

\bibitem{shen2017fplinq}
K.~Shen and W.~Yu, ``{FPLinQ}: A cooperative spectrum sharing strategy for
  device-to-device communications,'' in {\em Proc. IEEE Int. Symp. Info.
  Theory}, pp.~2323--2327, Aachen, Germany, Jun. 2017.

\bibitem{qian2012distributed}
L.~P. Qian, Y.~J.~A. Zhang, and M.~Chiang, ``Distributed nonconvex power
  control using {Gibbs} sampling,'' {\em IEEE Trans. Commun.}, vol.~60,
  pp.~3886--3898, Dec. 2012.

\bibitem{itu1411}
{\em Recommendation ITU-R P.1411-8}.
\newblock International Telecommunication Union, 2015.

\bibitem{insoo13holistic}
I.~Hwang, B.~Song, and S.~S. Soliman, ``A holistic view on hyper-dense
  heterogeneous and small cell networks,'' {\em IEEE Commun. Mag.}, vol.~51,
  pp.~20--27, Jun. 2013.

\bibitem{zaheer2017deep}
M.~Zaheer, S.~Kottur, S.~Ravanbakhsh, B.~Poczos, R.~Salakhutdinov, and
  A.~Smola, ``Deep sets,'' in {\em Proc. Adv. Neural Inform. Process. Syst.},
  pp.~1--9, Long Beach, CA, USA, Dec. 2017.

\end{thebibliography}

\end{document}